\documentclass[a4paper,11pt,reqno]{amsart}

\usepackage[margin=3cm]{geometry}

\usepackage[utf8]{inputenc}
\usepackage[T1]{fontenc}
\usepackage{dsfont}
\usepackage{mathtools}
\usepackage{enumerate}
\usepackage[inline]{enumitem}
\usepackage{tikz}
\usepackage{todonotes}
\usepackage{amssymb}
\usepackage{amsmath}
\usepackage{etoolbox}
\usepackage{booktabs}
\usepackage[font=footnotesize]{caption}
\usepackage{url}
\usepackage{xparse}
\usepackage{xspace}
\usepackage[foot]{amsaddr}
\usepackage[font=footnotesize]{subcaption}
\usepackage[titlenumbered]{algorithm2e}
\usepackage[colorlinks,citecolor=blue,linkcolor=blue,
pdftitle={Integrality of Linearizations of Polynomials over Binary Variables using Additional Monomials},
pdfauthor={Christopher Hojny, Marc E. Pfetsch, Matthias Walter}]{hyperref}

\usetikzlibrary{decorations.pathmorphing}
\usetikzlibrary{calc}

\DeclareDocumentCommand\linearization{}{\ensuremath{\mathcal{L}}\xspace}
\DeclareDocumentCommand\singletonMonomials{}{\ensuremath{\mathcal{S}}}
\DeclareDocumentCommand\properMonomials{}{\ensuremath{\mathcal{M}^{p}}}
\DeclareDocumentCommand\allMonomials{}{\ensuremath{\mathcal{M}}}
\DeclareDocumentCommand\targetMonomials{}{\ensuremath{\mathcal{T}}}
\DeclareDocumentCommand\proj{oo}{\IfValueTF{#1}{\text{proj\,}{}_{#1}}{%
   \text{proj\,}{}}\IfValueTF{#2}{\left(#2\right)}{}}

\DeclareDocumentCommand\constraints{}{\ensuremath{\mathcal{A}}\xspace}
\DeclareDocumentCommand\constraint{}{\ensuremath{\mathcal{C}}\xspace}
\DeclareDocumentCommand\upperNodes{}{\ensuremath{U}\xspace}
\DeclareDocumentCommand\lowerNodes{}{\ensuremath{L}\xspace}
\DeclareDocumentCommand\pred{}{\ensuremath{\mathop{pred}}\xspace}
\DeclareDocumentCommand\succ{}{\ensuremath{\mathop{succ}}\xspace}

\DeclareDocumentCommand\nonemptyIntersection{}{\ensuremath{\mathcal{I}}}
\DeclareDocumentCommand\poset{}{\ensuremath{\mathcal{L}^\star}}

\DeclareDocumentCommand\AND{}{\textsc{AND}\xspace}

\newcommand{\BQP}{\mathrm{BQP}}
\newcommand{\BMP}{\mathrm{BMP}}

\newcommand{\B}[1]{\{0,1\}^{#1}}
\newcommand{\R}{\mathds{R}}
\newcommand{\Z}{\mathds{Z}}

\newcommand{\Q}{\mathds{Q}}

\newcommand{\card}[1]{\lvert{#1}\rvert}

\newcommand{\define}{\coloneqq}

\newcommand{\suchthat}{\;\colon\;}

\DeclareDocumentCommand\setdef{mo}{\left\{#1\IfNoValueTF{#2}{}{ : #2}\right\}}
\DeclareDocumentCommand\transpose{m}{#1^{\top}}
\DeclareDocumentCommand\zerovec{o}{\IfNoValueTF{#1}{\mathbb{O}}{\mathbb{O}_{#1}}}
\DeclareMathOperator{\conv}{conv}
\DeclareMathOperator{\argmin}{argmin}

\DeclareMathOperator{\cl}{cl}

\DeclareDocumentCommand\orderO{o}{\mathcal{O}\IfValueTF{#1}{\left(#1\right)}{}}
\DeclareDocumentCommand\cplxNP{}{\mathsf{NP}}

\usepackage[capitalize,noabbrev]{cleveref}

\makeatletter
\renewcommand{\subsection}{\@startsection{subsection}{2}%
  {\z@}%
  {.7\linespacing\@plus\linespacing}%
  {.5\linespacing}%
  {\normalfont\scshape\centering}}
\makeatother

\makeatletter
\patchcmd{\@startsection}
  {\@afterindenttrue}
  {\@afterindentfalse}
  {}{}
\makeatother

\theoremstyle{definition}
\newtheorem{definition}{Definition}
\newtheorem{example}[definition]{Example}

\newtheorem{remark}[definition]{Remark}

\theoremstyle{plain}
\newtheorem{proposition}[definition]{Proposition}
\newtheorem{lemma}[definition]{Lemma}
\newtheorem{theorem}[definition]{Theorem}
\newtheorem{claim}[definition]{Claim}

\crefname{Claim}{claim}{claims}
\newtheorem{corollary}[definition]{Corollary}

\newcommand*{\claimproofname}{Proof}
\newenvironment{claimproof}[1][\claimproofname]{\begin{proof}[#1]}{\end{proof}}
\allowdisplaybreaks

\setlist[enumerate]{leftmargin=5ex}
\setlist[itemize]{label={$\circ$},leftmargin=3ex}

\SetKwInput{AlgoInput}{\textbf{Input}}
\SetKwInput{AlgoOutput}{\textbf{Output}}
\DontPrintSemicolon

\tikzstyle{nodeMonomial} += [circle,draw=black,thick,inner sep=0pt,minimum size=2mm]
\tikzstyle{nodeTarget} += [nodeMonomial,fill=black]
\tikzstyle{nodeSingleton} += [nodeMonomial,fill=black!30]
\tikzstyle{arc} += [->,thick]
\tikzstyle{path} += [arc,decorate,decoration={snake,amplitude=0.5pt}]

\title[Integral Linearizations of Binary Polynomials Using Additional Monomials]{Integrality of Linearizations of Polynomials over Binary Variables using Additional Monomials}
\author[Christopher Hojny]{Christopher Hojny$^1$}
\address{$^1$Combinatorial Optimization Group, Eindhoven University of
  Technology, Eindhoven, The Netherlands, {\normalfont{\textit{E-mail}: \texttt{c.hojny@tue.nl}}}}

\author[Marc E.~Pfetsch]{Marc E.~Pfetsch$^2$}
\address{$^2$Research Group Optimization, TU Darmstadt, Darmstadt, Germany, {\normalfont{\textit{E-mail}: \texttt{pfetsch@opt.tu-darmstadt.de}}}}

\author[Matthias Walter]{Matthias Walter$^3$}
\address{$^3$Department of Applied Mathematics, University of Twente,
  Enschede, The Netherlands, {\normalfont{\textit{E-mail}: \texttt{m.walter@utwente.nl}}}}

\date{November 2019}

\begin{document}

\begin{abstract}
  Polynomial optimization problems over binary variables can be expressed as integer programs using a linearization with extra monomials in addition to those arising in the given polynomial.
  We characterize when such a linearization yields an integral relaxation polytope, generalizing work by Del Pia and Khajavirad (SIAM Journal on Optimization, 2018) and Buchheim, Crama and Rodríguez-Heck (European Journal of Operations Research, 2019).
  We also present an algorithm that finds these extra monomials for a given polynomial to yield an integral relaxation polytope or determines that no such set of extra monomials exists.
  In the former case, our approach yields an algorithm to solve the given polynomial optimization problem as a compact LP, and we complement this with a purely combinatorial algorithm.
\end{abstract}
\maketitle

\keywords{\textbf{Keywords:} polynomial optimization, linearization, integer linear programming}
\bigskip

\noindent
\textcolor{red}{After the completion of this manuscript we learned that
  results proved here have essentially already been obtained by Del Pia
  and Khajavirad (``The running intersection relaxation of the multilinear
  polytope'', Optimization Online,
  \url{http://www.optimization-online.org/DB_HTML/2018/05/6618.html}). Note
  that we use a different notation.}

\section{Introduction}
\label{sec:intro}

For a positive integer $n$, the problem of minimizing a polynomial over binary variables \mbox{$x_1$, $x_2$, \dots, $x_n$} can be formulated as
\begin{equation}
   \label{eq:polyMin}
   \min ~\Big\{ \sum_{m \in \targetMonomials} a_m \cdot \prod_{i \in m} x_i \suchthat x \in \setdef{0,1}^n \Big\},
\end{equation}
where $\targetMonomials$ is a set of subsets of $[n] \define \setdef{1, \dotsc, n}$
representing \emph{target monomials} and $a_m \in \R$, $m \in \targetMonomials$, are the
corresponding monomial coefficients.
Note that due to $x^2 = x$ for binary variables~$x$, no higher order powers appear.

A common approach to solve~\eqref{eq:polyMin} is to linearize the monomials in order to obtain a binary linear program.
The most common way is the following.
For each monomial $m \in \targetMonomials$, one adds a binary variable $z_m$ representing the value of the monomimal, i.e.,
$z_m = \prod_{i \in m} x_i$.
One can view $z_m$ as the conjunction/\AND-resultant of the variables in
the monomial.
An integer formulation (the \AND-formulation/relaxation) for this relation is given by
\begin{equation}
  \label{eq:standardLinearization}
   \sum_{i \in m} x_i \leq \card{m} - 1 + z_m \qquad \text{ and } \qquad z_m \leq x_i \text{ for all } i \in m.
\end{equation}
For the special case for $|m| = 2$ the inequalities are sometimes called \emph{McCormick inequalities}~\cite{McCormick76}.
With these variables the objective function in~\eqref{eq:polyMin} then
becomes the linear function $\sum_{m \in \targetMonomials} a_m \cdot z_m$.
The resulting formulation is called the \emph{standard linearization}.
Del Pia and Khajavirad~\cite{DelPiaKhajavirad2018} and Buchheim, Crama and Rodríguez-Heck~\cite{BuchheimCR19} characterized when such a linearization yields integral polytopes.

Clearly, when constructing a linearization, one is not limited to using the target monomials.
For instance, for the target monomials $\targetMonomials = \{\{1,2,3\}, \{2,3,4\}\}$ one may add the auxiliary monomial $\{2,3\}$,
i.e., replace $z_{\{1,2,3\}} = x_1 \cdot x_2 \cdot x_3$ and $z_{\{2,3,4\}} = x_2 \cdot x_3 \cdot x_4$ by $z_{\{1,2,3\}} = x_1 \cdot z_{\{2,3\}}$, $z_{\{2,3,4\}} = z_{\{2,3\}} \cdot x_4$ and $z_{\{2,3\}} = x_2 \cdot x_3$.
Although this makes no difference for the integer solutions, it turns out that the strengths of the corresponding linearizations differ.
In this paper, we extend the integrality characterization of Del Pia and Khajavirad to the case with additional monomials.

\subsection{Literature Review}

Linearizations of multilinear expressions over binary variables have been widely studied in the literature.
To find the strongest linearization of a set of bilinear terms over a set of binary variables, Padberg~\cite{Padberg1989} introduced the \emph{Boolean Quadric Polytope}
\begin{equation*}
   \BQP_n \define \conv \setdef{ (x,y) \in \setdef{0,1}^n \times \setdef{0,1}^{\frac{n(n-1)}{2}} }[ x_i \cdot x_j = y_{i,j} \text{ for all } i, j \in [n] \text{ with } i < j ],
\end{equation*}
which allows to minimize a bilinear expression in binary variables by
integer linear programming techniques.
A standard linear programming formulation of~$\BQP_n$ is to replace~$x_i
\cdot x_j = y_{i,j}$ by the inequalities from the standard linearization~\eqref{eq:standardLinearization} for each term, i.e., adding the McCormick inequalities.
To strengthen this integer programming formulation of~$\BQP_n$, several facet defining inequalities have been derived, see, e.g., Boros and Hammer~\cite{BorosHammer1993}, Macambira and De Souza~\cite{MacambiraDeSouza20003}, and Yajima and Fujie~\cite{YajimaFujie1998}.
Boros et al.~\cite{BorosEtAl1992} completely characterize the first Chv\'atal closure of the standard formulation for $\BQP_n$.
Moreover, De Simone~\cite{DeSimone1990} showed that~$\BQP_n$ is the image of a cut polytope under a bijective linear transformation.
Thus, facet descriptions (or relaxations) of~$\BQP_n$ can be found by studying cut polytopes and applying the linear transformation.
Barahona et al.~\cite{BarahonaEtAl1989} use the relation of~$\BQP_n$ to the cut polytope to develop a branch-and-cut algorithm for solving quadratic 0-1 programs. Letchford and S\o{}rensen~\cite{LetchfordSorensen2014} present an efficient separation algorithm for certain classes of facet defining inequalities for~$\BQP_n$.

Analogously to the bilinear case, multilinear optimization problems can be solved by integer linear programming techniques by minimizing a linear function over the \emph{Boolean Multilinear Polytope}
\begin{equation*}
   \BMP_n(\targetMonomials) \define \conv \Big\{(x,y) \in \setdef{0,1}^n \times \setdef{0,1}^{\targetMonomials} \suchthat \prod_{i \in m} x_i = y_m \text{ for all } m \in \targetMonomials\Big\},
\end{equation*}
where~$\targetMonomials$ is a set of subsets of~$[n]$ encoding multilinear
expressions in a set of~$n$ binary variables.
Integer programming formulations for~$\BMP_n(\targetMonomials)$ are given by the inequalities in~\eqref{eq:standardLinearization} or similar relaxations as discussed by Watters~\cite{Watters1967} or Glover and Woolsey~\cite{GloverWoolsey1974}.
Techniques to reduce the number of constraints in a specific relaxation
of~$\BMP_n(\targetMonomials)$ are discussed by Glover and
Woolsey~\cite{GloverWoolsey1973}.
If~$\card{\targetMonomials} = 1$, the standard linearization is a complete linear description of~$\BMP_n(\targetMonomials)$, but if~$\card{\targetMonomials} \geq 2$, it is not necessarily complete.
Crama and Rodr\'iguez-Heck~\cite{CramaRodriguezHeck2017} derive a family of inequalities that completely describe $\BMP_n(\targetMonomials)$ together with inequalities from the \AND-linearization for every term in~$\targetMonomials$ if~$\card{\targetMonomials} \leq 2$.
Using a hypergraph model to encode~$\targetMonomials$, Del Pia and Khajavirad~\cite{DelPiaKhajavirad2017,DelPiaKhajavirad2018,DelPiaKhajavirad2018a} derive further facet defining inequalities of~$\BMP_n(\targetMonomials)$ by exploiting hypergraph substructures.
In particular, they show that the standard relaxation completely describes~$\BQP_n(\targetMonomials)$ if and only if the hypergraph associated with~$\targetMonomials$ is Berge-acyclic.
Balas and Mazzola~\cite{BalasMazzola1984} derive affine under- and overestimators for a multilinear expression without using the~$y$-variables of~$\BMP_n(\targetMonomials)$.
Furthermore, for a given multilinear inequality, they provide a set of linear inequalities having the same binary solutions.
Since this class of inequalities may be exponentially large, they also show dominance relations between these inequalities to reduce the number of necessary
inequalities~\cite{BalasMazzola1984a}.
To the best of our knowledge, the separation problem for these inequalities is $\cplxNP$-hard in general.

An alternative line of research is to introduce artificial variables for monomials that are not contained in~$\targetMonomials$, e.g., by splitting a monomial into several smaller ones and combining the variables of submonomials via \AND-constraints as indicated in the introduction.
This approach is widely used by practitioners and in solvers, see, e.g.,
Boros and Hammer~\cite{BorH02}, \cite{BerHP09}, or the Pseudo Boolean Competitions, e.g., \cite{PB16}.
Based on this idea, Buchheim and Rinaldi~\cite{BuchheimRinaldi2008} show how arbitrary multilinear programs can be reduced to the quadratic case by
introducing few new variables.
Alternatively, each monomial~$\setdef{i_1, \dots,
i_k }$ in~$\targetMonomials$ can be reduced iteratively by one variable by
introducing an auxiliary variable~$y$ that models the product~$x_{i_1}
\cdot \prod_{j = 2}^k x_{i_j}$ and using McCormick inequalities to link~$y$
with both factors of this product.
Since this relaxation does not exploit the structure of the set~$\targetMonomials$, the corresponding relaxation is typically weak.
Luedtke et~al.~\cite{LuedtkeEtAl2012} investigate the difference of this linearization and the concave and convex envelopes of multilinear expressions, i.e., the strongest relaxation.
In particular, they show that this difference can be arbitrarily large.

\subsection{Contribution and Outline}

Our main contribution is the analysis of polytopes of linearizations that use extra monomials in addition to the singleton and target monomials.
Preliminary results on such linearizations are derived in \cref{sec:linearizations}, in which we restrict ourselves to a subclass of linearizations that involve only a reasonable number of additional monomials (in contrast to all monomial resultants of possible \AND-constraints).
For such linearizations, an associated digraph turns out to be helpful for the analysis.
In \cref{sec:integrality}, our first main result is presented, which is an integrality characterization for the corresponding relaxation polytopes (\cref{thm:simpleLinearizationProjectionIntegral}).
The second main result, presented in \cref{sec:existence}, solves the problem of determining if extra monomials can be added in order to obtain an integral relaxation (\cref{thm:RIP}).
The latter problem is also addressed algorithmically in \cref{thm:RIPalgorithm}, which yields a polynomial-time algorithm to solve binary polynomial optimization problems, provided there exists a linearization whose digraph is acyclic in the undirected sense.
In the final section, we give an outlook on further research directions.

\section{Simple Linearizations and Their Digraphs}
\label{sec:linearizations}

We consider polynomials in $\R[x_1, x_2, \dotsc, x_n]$ with variables $x_1$, $x_2$,\dots, $x_n$.
As basic building blocks, we represent monomials by the index sets of their variables, i.e., $m \subseteq [n]$ encodes the monomial $m(x) \define \prod_{i \in m} x_i$.
The \emph{target monomials} $\targetMonomials$ are the monomials we are interested in, i.e., those appearing in the polynomial.
The variables $x_1, \dotsc, x_n$ are identified with the variables for the \emph{singleton monomials} $\singletonMonomials \define \setdef{ \setdef{i} }[ i \in [n] ]$.

A \emph{linearization}, denoted by $\linearization = (n, \allMonomials, \constraints)$, consists of a set $\allMonomials \subseteq 2^{[n]}$ of monomials as well as a set~\constraints of \AND-constraints.
In order to model optimization problems over polynomials with binary variables $x_1, \dotsc, x_n$ and target monomials $\targetMonomials$, we assume \mbox{$\allMonomials \supseteq \singletonMonomials \cup \targetMonomials$}.
Each \emph{\AND-constraint}~$\constraint \in \constraints$ combines monomials $m_1$, $m_2$, \dots, $m_k \in \allMonomials$ to a new monomial in $\allMonomials$.
This monomial evaluates to 1 if and only if all variables of the combined monomials are 1.
Thus, we write $\constraint \define \setdef{m_1, m_2, \dotsc, m_k}$ and the resulting monomial is $\bigcup \constraint \define m_1 \cup m_2 \cup \dotsb \cup m_k \in \allMonomials$.
Throughout our paper, we denote by $\properMonomials \define \allMonomials \setminus \singletonMonomials$ the \emph{proper monomials}.
Moreover, we only consider linearizations that are \emph{consistent} in the sense that each proper monomial $m \in \properMonomials$ is the resultant of an \AND-constraint, and that the monomials of that \AND-constraint are proper subsets of $m$.
In other words, there shall exist at least one constraint $\constraint \in \constraints$ with $m = \bigcup \constraint$ for each $m \in \properMonomials$ and $|m'| < |m|$ holds for each $m' \in \constraint$.

The \emph{standard linearization} mentioned in the introduction is then given by $\linearization = (n, \allMonomials, \constraints)$ with $\allMonomials = \singletonMonomials \cup \targetMonomials$ and $\constraints = \setdef{ \setdef{ \setdef{i} }[ i \in m ] }[ m \in \targetMonomials ]$.
It has one constraint per target monomial, combining the singleton monomials contained in this target monomial.

A linearization $\linearization = (n, \allMonomials, \constraints)$ induces a \emph{relaxation} $P(\linearization) \subseteq \R^{\allMonomials} $, which is the polytope defined by the following inequalities on the binary variables $y_m$ for all monomials $m \in \allMonomials$:
\begin{subequations}
  \label{EquationLinearization}
  \begin{alignat}{7}
    y_m &\in [0,1] &&\qquad \text{for all } m \in \allMonomials, \label{eq:linearizationBounds} \\
    y_{\bigcup \constraint} &\leq y_m &&\qquad \text{for all } \constraint \in \constraints \text{ and all } m \in \constraint \label{eq:linearizationPair}, \text{ and} \\
    \sum_{m \in \constraint} y_m &\leq y_{\bigcup \constraint} + \card{\constraint} - 1 &&\qquad \text{for all } \constraint \in \constraints. \label{eq:linearizationSum}
  \end{alignat}
\end{subequations}

It is well known that constraints~\eqref{eq:linearizationBounds}--\eqref{eq:linearizationSum} correctly model products of binary variables, which is a consequence of the following result, which we provide for completeness.

\begin{proposition}
   \label{thm:linearizationConsistency}
   Let $\linearization = (n, \allMonomials, \constraints)$ be a linearization, and let $\targetMonomials \subseteq \properMonomials$ be a subset of the proper monomials.
   A binary vector $y \in \setdef{0,1}^{\singletonMonomials \cup \targetMonomials}$ can be extended to $y' \in P(\linearization) \cap \Z^{\allMonomials}$ if and only if $y_m = \prod_{i \in m} y_{\setdef{i}}$ holds for each monomial $m \in \targetMonomials$.
\end{proposition}

\begin{proof}
   For sufficiency, we extend $y$ to $y'$ by setting $y'_m \define \prod_{i \in m} y_{\setdef{i}} \in \setdef{0,1}$ for all monomials $m \in \properMonomials$.
   Clearly, $y'$ extends $y$ since $y'_m = y_m$ holds for all $m \in \singletonMonomials \cup \targetMonomials$.
   Since $y'$ is   binary, it remains to verify that $y'$ satisfies constraints~\eqref{eq:linearizationPair} and~\eqref{eq:linearizationSum}.
   Let $\constraint \in \constraints$ be one of the constraints of $\linearization$ and let $m' \define \bigcup \constraint$ be the resulting monomial.
   For each $m \in \constraint$ we have $m' \supseteq m$, which proves $y'_{m'} = \prod_{i \in m'} y'_{\setdef{i}} \leq \prod_{i \in m} y'_{\setdef{i}} = y'_m$, establishing~\eqref{eq:linearizationPair}.
   To see that also~\eqref{eq:linearizationSum} is satisfied, assume, for the sake of contradiction, that $\sum_{m \in \constraint} y'_m > y'_{m'} + \card{\constraint} - 1$ holds.
   Since $y'$ is binary, this implies $y'_m = 1$ for all $m \in \constraint$.
   This in turn implies $y'_{\setdef{i}} = 1$ for all $i \in m'$, and we obtain $y'_{m'} = 1$, a contradiction.

   For necessity, let $y' \in P(\linearization) \cap \Z^{\allMonomials}$ extend a vector $y \in \setdef{0,1}^{\singletonMonomials \cup \targetMonomials}$.
   Assume, for the sake of contradiction, that there exists a proper monomial $m' \in \properMonomials$  such that $y'_{m'} \neq \prod_{i \in m'} y'_{\setdef{i}}$ holds.
   We can assume $m'$ to have minimal cardinality among all such monomials.
   By consistency of $\linearization$, there exists a constraint $\constraint \in \constraints$ with $m' = \bigcup \constraint$ such that $\card{m} < \card{m'}$ holds for each $m \in \constraint$.
   Hence, minimality of $\card{m'}$ ensures that $y'_m = \prod_{i \in m} y'_{\setdef{i}}$.
   It is easy to see that $y'_{m'} = \prod_{m \in \constraint} y'_m$ holds by constraints~\eqref{eq:linearizationBounds}--\eqref{eq:linearizationSum}.
   We obtain $y'_{m'} = \prod_{m \in \constraint} \prod_{i \in m} y'_{\setdef{i}} = \prod_{i \in m'} y'_{\setdef{i}}$, where the last equation holds due to $y'_{\setdef{i}} \cdot y'_{\setdef{i}} = y'_{\setdef{i}}$. This contradicts our assumption and concludes the proof.
\end{proof}

In principle, linearizations admit that a (proper) monomial is the result of several \AND-constraints.
However, for practical purposes (e.g., the number of constraints), it is interesting to consider those linearizations in which each proper monomial is the result of \emph{exactly one} \AND-constraint.
We call such linearizations \emph{simple}.
Note that this is equivalent to the requirement $\card{\constraints} = \card{\properMonomials}$.

The \emph{digraph} associated with a simple linearization $\linearization = (n, \allMonomials, \constraints)$ is a directed acyclic graph $D(\linearization)$ having nodes for all monomials $m \in \allMonomials$.
There is an arc from the node $\bigcup \constraint$ for each constraint $\constraint \in \constraints$ to each of its \emph{child nodes} $m \in \constraint$.
Note that a simple linearization is uniquely determined by its digraph.
The \emph{in-degree} of monomial $m \in \allMonomials$ is the number of \AND-constraints $\constraint \in \constraints$ with $m \in \constraint$.
Similarly, the \emph{out-degree} of a monomial~$\bigcup \constraint \in \properMonomials$ is $\card{\constraint}$ and $0$ otherwise.
A first simple property that can be expressed in terms of $D(\linearization)$ is the following implication of consistency.
\begin{proposition}
   \label{thm:consistencyPaths}
   Let $\linearization = (n, \allMonomials, \constraints)$ be a linearization and let $m \in \allMonomials$ be a monomial.
   Then for each $i \in m$, the digraph $D(\linearization)$ contains at least one path from node $m$ to the singleton node $\{i\} \in \singletonMonomials$.
\end{proposition}

We close this section with an example that illustrates the structure of a
linearization's digraph. We will use this example as a running example
throughout this article to illustrate different properties of
linearizations.

\begin{example}
  \label{ex:runningExample}
  Let~$n = 6$ and consider the set of target monomials
  \[
  \targetMonomials = \big\{ \{1,2,3,4\}, \{3,4,5\}, \{4,5,6\}\big\}.
  \]
  Figure~\ref{fig:runningExample} shows the digraph of a simple linearization
  of~$\targetMonomials$, where singleton monomials are colored gray and
  target monomials are colored black.
  Throughout this article, we use the linearization corresponding to this digraph as a running example.
\end{example}

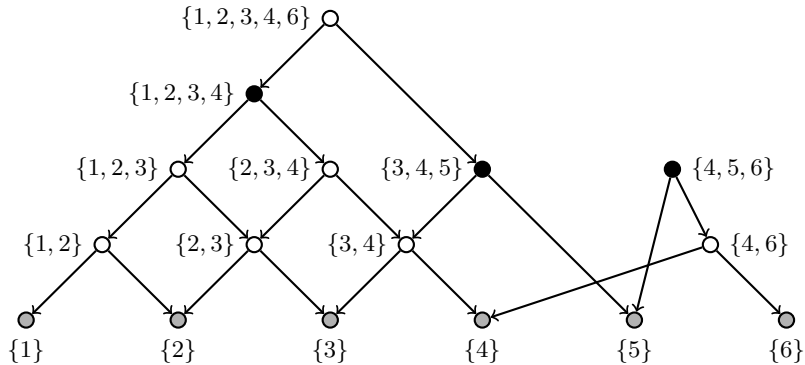
\begin{figure}[htb]
   \begin{tikzpicture}
      \node[nodeSingleton,label=below:\footnotesize{$\{1\}$}] (1) at (0,0) {};
      \node[nodeSingleton,label=below:\footnotesize{$\{2\}$}] (2) at (2,0) {};
      \node[nodeSingleton,label=below:\footnotesize{$\{3\}$}] (3) at (4,0) {};
      \node[nodeSingleton,label=below:\footnotesize{$\{4\}$}] (4) at (6,0) {};
      \node[nodeSingleton,label=below:\footnotesize{$\{5\}$}] (5) at (8,0) {};
      \node[nodeSingleton,label=below:\footnotesize{$\{6\}$}] (6) at (10,0) {};

      \node[nodeMonomial,label=left:\footnotesize{$\{1,2\}$}] (12) at (1,1) {};
      \node[nodeMonomial,label=left:\footnotesize{$\{2,3\}$}] (23) at (3,1) {};
      \node[nodeMonomial,label=left:\footnotesize{$\{3,4\}$}] (34) at (5,1) {};
      \node[nodeMonomial,label=right:\footnotesize{$\{4,6\}$}] (46) at (9,1) {};

      \node[nodeMonomial,label=left:\footnotesize{$\{1,2,3\}$}] (123) at (2,2) {};
      \node[nodeMonomial,label=left:\footnotesize{$\{2,3,4\}$}] (234) at (4,2) {};
      \node[nodeTarget,label=left:\footnotesize{$\{3,4,5\}$}] (345) at (6,2) {};
      \node[nodeTarget,label=right:\footnotesize{$\{4,5,6\}$}] (456) at (8.5,2) {};

      \node[nodeTarget,label=left:\footnotesize{$\{1,2,3,4\}$}] (1234) at (3,3) {};

      \node[nodeMonomial,label=left:\footnotesize{$\{1,2,3,4,6\}$}] (12345) at (4,4) {};

      \foreach \a/\b in {%
         12345/1234, 12345/345,
         1234/123, 1234/234,
         123/12, 123/23,
         234/23, 234/34,
         345/34, 345/5,
         456/46, 456/5,
         12/1, 12/2,
         23/2, 23/3,
         34/3, 34/4,
         46/4, 46/6%
         }%
      {
         \draw[arc] (\a) -- (\b);
      }
   \end{tikzpicture}
   \caption{The linearization graph of a simple linearization corresponding to the target monomials from \cref{ex:runningExample}.}
   \label{fig:runningExample}
\end{figure}

\section{Integrality of Linearizations}
\label{sec:integrality}

We say that a linearization $\linearization = (n, \allMonomials, \constraints)$ is \emph{integral with respect to $\targetMonomials \subseteq \properMonomials$} if the projection of its relaxation
$P(\linearization)$ onto the variables $y_m$ with $m \in \singletonMonomials \cup \targetMonomials$ is an integral polytope.
In this section, we will prove our first main result, a characterization of this property for simple linearizations in terms of $D(\linearization)$.
In fact, if the inequality system~\eqref{EquationLinearization} defines an integral polytope for a given linearization, then it is totally dual integrality (TDI).

We first consider a single \AND-constraint and establish the TDI property.
Then we provide a preprocessing algorithm that removes redundant parts of a linearization.
This algorithm motivates our integrality conditions whose sufficiency and necessity are proved in the last part of this section.

\subsection{A Single \texorpdfstring{\AND}{AND}-constraint}
\label{sec:ntegralitySingle}

It is well known that~\eqref{EquationLinearization} defines an integral polytope (see~\cite{Al-KhayyalF83,McCormick76} for the bilinear case and \cite{Crama93,LuedtkeEtAl2012} for the general case).
To the best of our knowledge, however, it is not widely known that the corresponding system is TDI.

\begin{proposition}
   \label{thm:ANDisTDI}
   Let~$m$ be a monomial, and let~$\linearization$ be the linearization that consists of the single \AND-constraint for~$m$.
   Then System~\eqref{EquationLinearization} is TDI, and thus~$P(\linearization)$ is an integral polytope.
\end{proposition}

\begin{proof}
   Let~$m = [n]$ and consider the standard linearization~$\linearization = (n, \singletonMonomials \cup \setdef{m}, \setdef{ \singletonMonomials })$.
   Let $(w, \bar{w}) \in \Z^n \times \Z$ be an arbitrary integral objective vector for~\eqref{EquationLinearization}.
   The dual of the corresponding maximization problem is
   \begin{align*}
    &&\min\quad (n-1) \beta + \sum_{i = 1}^n \gamma_i + \delta &&&\\
    &&-\alpha_i + \beta + \gamma_i &\geq w_i, &&& \text{ for all } i \in
    [n],\\
    &&\sum_{i = 1}^n \alpha_i - \beta + \delta &\geq \bar{w}, &&&\\
    &&\alpha_i,\, \beta,\, \gamma_i,\, \delta &\geq 0, &&& \text{ for all } i \in [n],
   \end{align*}
   where~$\alpha_i$ corresponds to the~$i$-th constraint of  type~\eqref{eq:linearizationPair}, $\beta$ to~\eqref{eq:linearizationSum}, and~$\gamma_i$ and~$\delta$ to the upper bounds of~$y_{\setdef{i}}$ and the resultant~$y_m$, respectively.
   We will construct feasible integral solutions for~\eqref{EquationLinearization} and its dual with the same objective value.
   To this end, define $S \define \setdef{ i \in [n] }[ w_i \geq 0 ]$ and choose $k \in \argmin\setdef{ w_i }[ i \in [n] ]$.
   All dual solutions will satisfy $\delta = 0$, and thus we omit it subsequently.
   We have the following case distinction on $(w, \bar{w})$ corresponding to different types of optimal solutions.
   \medskip

   \noindent
   \textbf{Case 1:} $S = [n]$ and $w_k + \bar{w} \leq 0$. \\
   Consider the primal solution with $y_{\setdef{k}} = y_m = 0$, $y_{\setdef{i}} = 1$ for all $i \in [n] \setminus \setdef{k}$, and the dual solution with $\alpha_i = 0$ and $\gamma_i = w_i - w_k$ for all $i \in [n]$ and $\beta = w_k$.
   Note for dual feasibility that $\sum_{i \in [n]} \alpha_i - \beta = 0 - w_k \geq \bar{w}$.
   Moreover, the common objective value is $(n-1) \beta + \sum_{i \in [n]} \gamma_i = (n-1) w_k + \sum_{i \in [n]} w_i - n \cdot w_k = \sum_{i \in [n]} w_i - w_k$.
   \medskip

   \noindent
   \textbf{Case 2:} $S = [n]$ and $w_k + \bar{w} > 0$. \\
   Consider the primal solution with $y_{\setdef{i}} = 1$ for all $i \in [n]$, $y_m = 1$, and the dual solution with $\alpha_i = 0$ and $\gamma_i = w_i + \min\setdef{\bar{w}, 0}$ for all $i \in [n] \setminus \setdef{k}$, $\alpha_k = \max\setdef{\bar{w}, 0}$, $\gamma_k = w_k + \bar{w}$ and $\beta = \max\setdef{-\bar{w}, 0}$.
   Note for dual feasibility that $\sum_{i \in [n]} \alpha_i - \beta = \max\setdef{\bar{w},0} - \max\setdef{-\bar{w}, 0} = \bar{w}$.
   Moreover, the common objective value is $(n-1) \beta + \sum_{i \in [n]} \gamma_i = (n-1) \max\setdef{-\bar{w},0} + \sum_{i \in [n]} w_i + (n-1) \min\setdef{\bar{w},0} + \bar{w} = \sum_{i \in [n]} w_i + \bar{w}$.
   \medskip

   \noindent
   \textbf{Case 3:} $S \neq [n]$, $\sum_{i \in [n] \setminus S} w_i + \bar{w} \leq 0$. \\
   Consider the primal solution with $y_{\setdef{i}} = 1$ for all $i \in S$, $y_{\setdef{i}} = 0$ for all $i \in [n] \setminus S$, $y_m = 0$, and the dual solution with $\alpha_i = \max\setdef{-w_i,0}$ and $\gamma_i = \max\setdef{w_i, 0}$ for all $i \in [n]$ and $\beta = 0$.
   Note for dual feasibility that $\sum_{i \in [n]} \alpha_i - \beta = \sum_{i \in [n] \setminus S} (-w_i) - 0 \geq \bar{w}$.
   Moreover, the common objective value is $(n-1) \beta + \sum_{i \in [n]} \gamma_i = \sum_{i \in S} w_i$.
   \medskip

   \noindent
   \textbf{Case 4:} $S \neq [n]$ and $\sum_{i \in [n] \setminus S} w_i + \bar{w} > 0$. \\
   Consider the primal solution with $y_{\setdef{i}} = 1$ for all $i \in [n]$, $y_m = 1$, and the dual solution with $\alpha_k = \bar{w} - w_k + \sum_{i \in [n] \setminus S} w_i$, $\gamma_k = \bar{w} + \sum_{i \in [n] \setminus S} w_i$, $\alpha_i = 0$ and $\gamma_i = w_i$ for all $i \in S$, $\alpha_i = -w_i$ and $\gamma_i = 0$ for $i \in [n] \setminus (S \cup \setdef{k})$ as well as $\beta = 0$.
   Note for dual feasibility that $\sum_{i \in [n]} \alpha_i - \beta = (\bar{w} - w_k + \sum_{i \in [n] \setminus S} w_i) + \sum_{i \in [n] \setminus (S \cup \setdef{k})} (-w_i) = \bar{w}$.
   Moreover, the common objective value is $(n-1) \beta + \sum_{i \in [n]} \gamma_i = 0 + \bar{w} + \sum_{i \in [n] \setminus S} w_i + \sum_{i \in S} w_i = \sum_{i \in [n]} w_i + \bar{w}$.
   \medskip

   This establishes TDI in every case and concludes the proof.
\end{proof}

\begin{remark}
   In the proof of \cref{thm:ANDisTDI}, each constructed dual solution has $\delta = 0$ and satisfies the constraints of the dual that correspond to variables $y_{\setdef{i}}$ with equality.
   Hence, the system~\eqref{EquationLinearization} for a single \AND-constraint is TDI even after removing the redundant constraints $y_{\bigcup \constraint} \leq 1$ and $y_{\setdef{i}} \geq 0$ for all $i \in [n]$.
\end{remark}

\begin{remark}
   For a single \AND-constraint, the constraint matrix of system~\eqref{EquationLinearization} is totally unimodular if and only if $\card{\constraint} = 2$.
   For $\card{\constraint} \geq 3$, it contains the submatrix
   \[
      \begin{pmatrix*}[r] 1 & -1 & 0 & 0 \\ 1 & 0 & -1 & 0 \\ 1 & 0 & 0 & -1 \\ -1 & 1 & 1 &1 \end{pmatrix*}
   \]
   with determinant $2$.
\end{remark}

\subsection{Preprocessing of Linearizations}
\label{sec:integralityPreprocessing}

A natural operation is the elimination of an auxiliary monomial.
Its implications on the induced relaxations are stated in the next lemma; see \cref{fig:preprocessing} for an illustration.
To this end, for~$M \subseteq \allMonomials$, we denote by~$\proj_{M}(P(\linearization))$ the orthogonal projection of~$P(\linearization)$ onto the variables~$y_m$ with~$m \in M$.

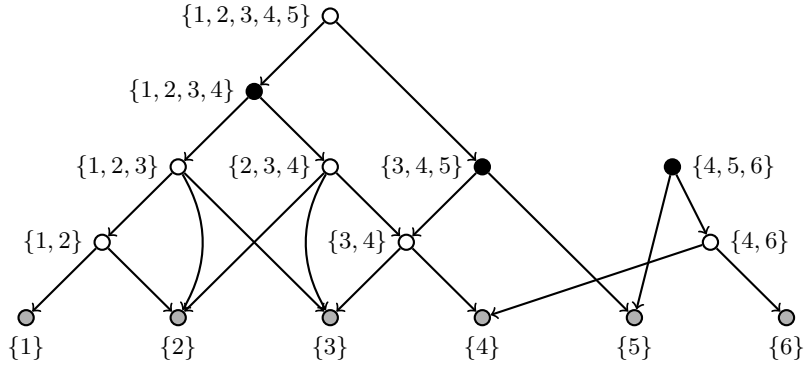
\begin{figure}
   \centering
   \begin{tikzpicture}
      \node[nodeSingleton,label=below:\footnotesize{$\{1\}$}] (1) at (0,0) {};
      \node[nodeSingleton,label=below:\footnotesize{$\{2\}$}] (2) at (2,0) {};
      \node[nodeSingleton,label=below:\footnotesize{$\{3\}$}] (3) at (4,0) {};
      \node[nodeSingleton,label=below:\footnotesize{$\{4\}$}] (4) at (6,0) {};
      \node[nodeSingleton,label=below:\footnotesize{$\{5\}$}] (5) at (8,0) {};
      \node[nodeSingleton,label=below:\footnotesize{$\{6\}$}] (6) at (10,0) {};

      \node[nodeMonomial,label=left:\footnotesize{$\{1,2\}$}] (12) at (1,1) {};
      \node[nodeMonomial,label=left:\footnotesize{$\{3,4\}$}] (34) at (5,1) {};
      \node[nodeMonomial,label=right:\footnotesize{$\{4,6\}$}] (46) at (9,1) {};

      \node[nodeMonomial,label=left:\footnotesize{$\{1,2,3\}$}] (123) at (2,2) {};
      \node[nodeMonomial,label=left:\footnotesize{$\{2,3,4\}$}] (234) at (4,2) {};
      \node[nodeTarget,label=left:\footnotesize{$\{3,4,5\}$}] (345) at (6,2) {};
      \node[nodeTarget,label=right:\footnotesize{$\{4,5,6\}$}] (456) at (8.5,2) {};

      \node[nodeTarget,label=left:\footnotesize{$\{1,2,3,4\}$}] (1234) at (3,3) {};

      \node[nodeMonomial,label=left:\footnotesize{$\{1,2,3,4,5\}$}] (12345) at (4,4) {};

      \foreach \a/\b in {%
         12345/1234, 12345/345,
         1234/123, 1234/234,
         123/12, 123/3,
         234/2, 234/34,
         345/34, 345/5,
         456/46, 456/5,
         12/1, 12/2,
         34/3, 34/4,
         46/4, 46/6%
      }%
      {
         \draw[arc] (\a) -- (\b);
      }
      \draw[arc] (123) to[bend left,bend angle=45] (2);
      \draw[arc] (234) to[bend right,bend angle=45] (3);
   \end{tikzpicture}
   \caption{Example of the reduction step in \cref{thm:fourierMotzkin} applied to monomial $m^\star = \setdef{2,3}$ of the linearization from \cref{ex:runningExample}.}
   \label{fig:preprocessing}
\end{figure}

\begin{lemma}\label{thm:fourierMotzkin}
   Let $\linearization = (n, \allMonomials, \constraints)$ be a simple linearization.
   Consider a proper monomial $m^\star \in \properMonomials$ and the linearization $\linearization' = (n, \allMonomials \setminus \setdef{m^\star}, \constraints')$ for which $D(\linearization')$ arises from $D(\linearization)$ by replacing every path of length $2$ containing $m^\star$ as inner node by a single arc from the start of the path to its end, and finally removing $m^\star$ and its remaining incident arcs.
   Then $\proj[\allMonomials \setminus \setdef{m^\star}][P(\linearization)] \subseteq P(\linearization')$, where equality holds if the in-degree of $m^\star$ in $D(\linearization)$ is at most $1$.
\end{lemma}

\begin{proof}
  Let $\constraint_1, \dotsc, \constraint_k$ be all the constraints that contain $m^\star$.
  Since $\linearization$ is simple and consistent, there exists a unique constraint $\constraint^\star$  with $m^\star = \bigcup \constraint^\star$.
  In $\linearization'$, the constraints $\constraint_i$ are replaced with $\constraint'_i = \constraint^\star \cup \constraint_i \setminus \setdef{m^\star}$ for $i \in [k]$ and $\constraint^\star$ is removed.
  The only inequalities from~\eqref{EquationLinearization} that constrain~$y_{m^\star}$ are
  \begin{align*}
    0 &\leq y_{m^\star},  \\
    y_{\bigcup \constraint_i} &\leq y_{m^\star}       && \text{ for all } i \in [k], \\
    \sum_{m \in \constraint^\star} y_m + 1 - \card{\constraint^\star} &\leq y_{m^\star}, \\
    y_{m^\star} &\leq 1,                          && \\
    y_{m^\star} &\leq y_m                         && \text{ for all } m \in \constraint^\star, \text{ and } \\
    y_{m^\star} &\leq y_{\bigcup \constraint_i} + \card{\constraint_i} - 1 - \sum_{m \in \constraint_i \setminus \setdef{m^\star}} y_m && \text{ for all } i \in [k].
  \end{align*}
  By Fourier-Motzkin elimination~\cite{Dines19,Fourier27,Motzkin36}, $\proj[\allMonomials \setminus \setdef{m^\star}][P(\linearization)]$ is obtained by keeping all inequalities of $P(\linearization)$ that do not involve $y_{m^\star}$ and by combining each of the first three inequalities above with each of the last three inequalities.

   First, the combination of $y_{m^\star} \geq 0$ with one of the last three inequalities yields a redundant inequality.
   To see this for the combination with the last inequality (for some $i \in [k]$), observe that $y_{\bigcup \constraint_i} + \card{\constraint_i} - 1 - \sum_{m \in \constraint_i \setminus \setdef{m^\star}} y_m \geq 0 + \card{\constraint_i} - 1 - \card{\constraint_i \setminus \setdef{m^\star}} = 0$ holds due to the bound inequalities of the variables that appear.
   Second, the combination of $y_{m^\star} \leq 1$ with one of the first three inequalities yields a redundant inequality as well.
   To see this for the combination with the third inequality, observe that $\sum_{m \in \constraint^\star} y_m + 1 - \card{\constraint^\star} \leq 1$ holds due to the bound inequalities of the variables that appear.
   Third, the combined inequality
   \[
      \sum_{m' \in \constraint^\star} y_{m'} + 1 - \card{\constraint^\star} \leq y_m
   \]
   for $m \in \constraint^\star$ is implied by inequalities $y_{m'} \leq 1$ for all $m' \in \constraint^\star \setminus \setdef{m}$. 
   Fourth, the combined inequality $y_{\bigcup \constraint_i} \leq y_m$ for $i \in [k]$ and $m \in \constraint^\star$ corresponds to inequality~\eqref{eq:linearizationPair} for $\linearization'$, and the combined inequality 
   \begin{align*}
      && \sum_{m \in \constraint^\star} y_m + 1 - \card{\constraint^\star}
         &\leq y_{\bigcup \constraint_i} + \card{\constraint_i} - 1 - \sum_{m \in \constraint_i \setminus \setdef{m^\star}} y_m && \\
      &\iff& \sum_{m \in \constraint^\star \cup \constraint_i \setminus \setdef{m^\star}} y_m
         &\leq y_{\bigcup \constraint_i} + \card{\constraint^\star} + \card{\constraint_i} - 2 \\
      &\iff& \sum_{m \in \constraint'_i} y_m
         &\leq y_{\bigcup \constraint'_i} + \card{\constraint'_i} - 1
   \end{align*}
   corresponds to inequality~\eqref{eq:linearizationSum} for
   $\linearization'$ since~$y_{\bigcup \constraint_i}$ is replaced by~$y_{\bigcup \constraint'_i}$.
   This shows that all constraints in~\eqref{EquationLinearization} for $\linearization'$ are implied. This establishes $\proj[\allMonomials \setminus \setdef{m^\star}][P(\linearization)] \subseteq P(\linearization')$.
   The only remaining combined inequality is
   \[
      y_{\bigcup \constraint_i} \leq y_{\bigcup \constraint_j} + \card{\constraint_j} - 1 - \sum_{m \in \constraint_j \setminus \setdef{m^\star}} y_m
   \]
   for $i$, $j \in [k]$.
   For $i = j$, it is equivalent to the sum of the inequalities $y_m \leq 1$ for all $m \in \constraint_i \setminus \setdef{m^\star}$, and thus redundant.
   This proves $\proj[\allMonomials \setminus \setdef{m^\star}][P(\linearization)] = P(\linearization')$ for $k \leq 1$ since $i = j$ must hold (if such constraints exist at all) in this case.
   This proves the second statement of the lemma.
\end{proof}

One application of \cref{thm:fourierMotzkin} is a preprocessing step that eliminates redundant parts of a simple linearization if the set of target monomials is a proper subset of~$\properMonomials$.
To this end, given a digraph $D = (V,A)$ and~$W \subseteq V$, $\succ(W)$ shall denote the set of nodes reachable from any node in~$W$ including the nodes from~$W$.
For $w \in V$ we abbreviate $\succ(w) \define \succ(\setdef{w})$.
Similarly, $\pred(W)$ (resp.\ $\pred(w)$) shall denote the set of nodes from which one can reach a node in $W$ (resp.\ node~$w$) including~$W$ (resp.~$w$).

\begin{corollary}
   \label{thm:redundantMonomials}
   Let $\linearization = (n, \allMonomials, \constraints)$ be a simple linearization, and let $\targetMonomials \subseteq \properMonomials$ be a subset of the proper monomials.
   Then
   \begin{equation*}
      \linearization' \define (n, \allMonomials', \constraints') \text{ with } 
      \allMonomials' \define \singletonMonomials \cup \succ(\targetMonomials) \text{ and }
      \constraints' \define \setdef{ \constraint \in \constraints }[ \bigcup \constraint \in \allMonomials' ]
   \end{equation*}
   is a simple linearization that satisfies $P(\linearization') = \proj[\allMonomials'][P(\linearization)]$.
   Moreover, $\linearization'$ can be computed in linear time in the size of $\linearization$.
\end{corollary}

\begin{proof}
   It is easy to check that $\linearization'$ is indeed a linearization.
   The fact that it is simple is inherited from $\linearization$.
   Moreover, the projection of $P(\linearization')$ onto the variables $\allMonomials'$ is equal to $P(\linearization)$ by \cref{thm:fourierMotzkin}, since we can carry out the removal of monomials sequentially, always considering a monomial $m^{\star} \notin \succ(\targetMonomials)$ that has in-degree~$0$.
   Finally, $\linearization'$ can be computed in linear time, since $\succ(\targetMonomials)$ can be computed by breadth-first search in $D(\linearization)$.
\end{proof}

By \cref{thm:linearizationConsistency}, a simple linearization~$\linearization = (n, \allMonomials, \constraints)$ yields an integer programming formulation of~$\BMP_n(\targetMonomials)$ with auxiliary variables.
If we are interested in deciding whether the relaxation in fact projects to a complete linear description of~$\BMP_n(\targetMonomials)$, \cref{thm:redundantMonomials} shows that it is sufficient to consider linearizations in which every monomial is a successor of a target monomial.
Thus, we can avoid unnecessarily complicated linearizations.

\subsection{Characterization of Integrality}
\label{sec:integralityCharacterization}

Let~$G(D)$ denote the undirected version of a digraph~$D$.
Moreover, for a digraph~$D$ and a graph~$G$, we denote their node sets by~$V(D)$ and~$V(G)$, respectively.
In order to characterize integral simple linearizations, we consider subgraphs $Z$ of $D$ whose underlying undirected graph~$G(Z)$ is a cycle.
Note that in this article, cycles are always \emph{simple}, i.e., each node appears in exactly two edges of the cycle.

We can now state our first main result.

\begin{theorem}
   \label{thm:simpleLinearizationProjectionIntegral}
   Let~$\linearization = (n, \allMonomials, \constraints)$ be a simple linearization, and let~$\targetMonomials \subseteq \properMonomials$ be a subset of the proper monomials.
   Then~$\proj[\singletonMonomials \cup \targetMonomials][P(\linearization)]$ is integral if and only if~$D(\linearization)$ does not contain a subgraph~$Z$ that satisfies
   \begin{enumerate}[label=(\alph{*})]
   \item
      \label{PropertyCycle}
      $G(Z)$ is a cycle, and
   \item
      \label{PropertyPathAbove}
      $V(Z) \subseteq \succ(\targetMonomials)$.
   \end{enumerate}
   Moreover, if the projection is integral and $\properMonomials \subseteq \succ(\targetMonomials)$ holds, then system~\eqref{eq:linearizationBounds}--\eqref{eq:linearizationSum} is TDI.
\end{theorem}

\begin{figure}[tb]
   \begin{subfigure}[t]{0.45\textwidth}
      \centering
      \begin{tikzpicture}
         \node (1) at (-1,0) [nodeSingleton,label=below:\footnotesize{$\{1\}$}] {};
         \node (2) at (1,0) [nodeSingleton,label=below:\footnotesize{$\{2\}$}] {};
         \node (3) at (2,0) [nodeSingleton,label=below:\footnotesize{$\{3\}$}] {};
         \node (4) at (3,0) [nodeSingleton,label=below:\footnotesize{$\{4\}$}] {};
         \node (5) at (4,0) [nodeSingleton,label=below:\footnotesize{$\{5\}$}] {};

         \node (13) at (-0.5,1.5) [nodeTarget,label=left:\footnotesize{$\{1,3\}$}] {};
         \node (12) at (0.1,0.8) [nodeTarget,label=below:\footnotesize{$\{1,2\}$}] {};
         \node (24) at (2.5,1) [nodeMonomial,label=right:\footnotesize{$\{2,4\}$}] {};
         \node (234) at (2.5,2) [nodeMonomial,label=left:\footnotesize{$\{2,3,4\}$}] {};
         \node (2345) at (3.2,3) [nodeTarget,label=left:\footnotesize{$\{2,3,4,5\}$}] {};

         \draw[arc, dashed, ultra thick] (13) -- (1);
         \draw[arc, dashed, ultra thick] (12) to[bend right] (1);
         \draw[arc, dashed,ultra thick] (12) to[bend left] (2);
         \draw[arc, dashed,ultra thick] (13) to[bend left] (3);
         \draw[arc, dashed,ultra thick] (24) -- (2);
         \draw[arc] (24) -- (4);
         \draw[arc, dashed,ultra thick] (234) -- (24);
         \draw[arc, dashed,ultra thick] (234) to[bend right] (3);
         \draw[arc] (2345) -- (5);
         \draw[arc] (2345) -- (234);
      \end{tikzpicture}
      \caption{%
         The dashed cycle destroys integrality according to \cref{thm:simpleLinearizationProjectionIntegral}.
         Note that $\setdef{2,3,4}$ and $\setdef{2,4}$ are in $\succ(\targetMonomials)$
         due to $\setdef{2,3,4,5} \in \targetMonomials$.
      }
      \label{fig:simpleLinearizationProjectionIntegralEvil}
   \end{subfigure}
   \hspace{1mm}
   \begin{subfigure}[t]{0.45\textwidth}
      \centering
      \begin{tikzpicture}
         \node (1) at (0,0) [nodeSingleton,label=below:\footnotesize{ $\{1\}$}] {};
         \node (2) at (1,0) [nodeSingleton,label=below:\footnotesize{ $\{2\}$}] {};
         \node (3) at (2,0) [nodeSingleton,label=below:\footnotesize{ $\{3\}$}] {};
         \node (4) at (3,0) [nodeSingleton,label=below:\footnotesize{ $\{4\}$}] {};
         \node (5) at (4,0) [nodeSingleton,label=below:\footnotesize{ $\{5\}$}] {};
         \node (6) at (5,0) [nodeSingleton,label=below:\footnotesize{ $\{6\}$}] {};

         \node (23) at (1.5,1) [nodeTarget,label=left:\footnotesize{ $\{2,3\}$}] {};
         \node (123) at (0.5,2) [nodeTarget,label=above:\footnotesize{ $\{1,2,3\}$}] {};
         \node (34) at (2.5,1) [nodeMonomial,label=right:\footnotesize{ $\{3,4\}$}] {};
         \node (234) at (2.0,2) [nodeMonomial,label=above:\footnotesize{ $\{2,3,4\}$}] {};
         \node (56) at (4.5,1) [nodeMonomial,label=right:\footnotesize{ $\{5,6\}$}] {};
         \node (3456) at (3.5,2) [nodeTarget,label=above:\footnotesize{ $\{3,4,5,6\}$}] {};

         \draw[arc] (23) -- (2);
         \draw[arc, dashed, ultra thick] (23) -- (3);
         \draw[arc, dashed, ultra thick] (34) -- (3);
         \draw[arc] (34) -- (4);
         \draw[arc] (56) -- (5);
         \draw[arc] (56) -- (6);
         \draw[arc] (123) -- (1);
         \draw[arc] (123) -- (23);
         \draw[arc, dashed, ultra thick] (234) -- (23);
         \draw[arc, dashed, ultra thick] (234) -- (34);
         \draw[arc] (3456) -- (34);
         \draw[arc] (3456) -- (56);
      \end{tikzpicture}
      \caption{%
         The only (undirected) cycle is dashed.
         Due to $\setdef{2,3,4} \notin \succ(\targetMonomials)$, the projection onto the variables for $\singletonMonomials$ and $\targetMonomials$ is integral.
      }
      \label{fig:simpleLinearizationProjectionIntegralNice}
   \end{subfigure}
   \caption{%
      The role of cycles in \cref{thm:simpleLinearizationProjectionIntegral}.
      Black nodes are those in $\targetMonomials$.
      Both graphs contain exactly one cycle.
      However, the projection is not integral for the linearization in \cref{fig:simpleLinearizationProjectionIntegralEvil}, while it is integral for that in \cref{fig:simpleLinearizationProjectionIntegralNice}.
      }
    \label{fig:simpleLinearizationIntegral}
\end{figure}
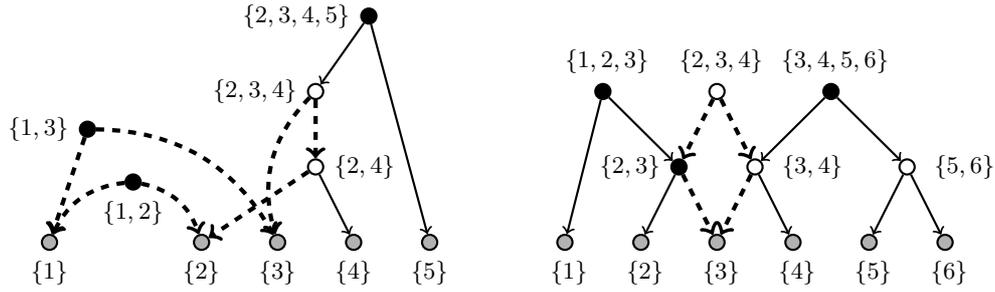

The role of the subgraph $Z$ in \cref{thm:simpleLinearizationProjectionIntegral} is illustrated in \cref{fig:simpleLinearizationIntegral}.
As an immediate consequence we obtain the following result.

\begin{corollary}
   \label{thm:simpleLinearizationIntegral}
   Let~$\linearization = (n, \allMonomials, \constraints)$ be a simple linearization.
   Then the following statements are equivalent:
   \begin{enumerate}[label=(\roman{*})]
   \item
      \label{thm:simpleLinearizationIntegralIntegrality}
      $P(\linearization)$ is integral.
   \item
      \label{thm:simpleLinearizationIntegralAcyclicity}
      $G(D(\linearization))$ is acyclic.
   \item
      \label{thm:simpleLinearizationIntegralTDI}
      System~\eqref{eq:linearizationBounds}--\eqref{eq:linearizationSum} is TDI.
   \end{enumerate}
\end{corollary}

\begin{proof}
   By choosing $\targetMonomials \define \properMonomials$ we have $\proj[\singletonMonomials \cup \targetMonomials][P(\linearization)] = P(\linearization)$.
   In this case, Property~\ref{PropertyPathAbove} of Theorem~\ref{thm:simpleLinearizationProjectionIntegral} is always satisfied due to $V(Z) \setminus \singletonMonomials \subseteq \properMonomials = \targetMonomials \subseteq \succ(\targetMonomials)$.
   The theorem then implies the equivalence of~\ref{thm:simpleLinearizationIntegralIntegrality} and~\ref{thm:simpleLinearizationIntegralAcyclicity} and the implication of~\ref{thm:simpleLinearizationIntegralTDI} by~\ref{thm:simpleLinearizationIntegralIntegrality}.
   The reverse implication is clear due to the integral right-hand side of system~\eqref{eq:linearizationBounds}--\eqref{eq:linearizationSum}.
\end{proof}

\begin{remark}
   Another consequence of \cref{thm:simpleLinearizationProjectionIntegral} is that we can test in linear time whether a given simple linearization is integral, or, more generally, whether a given orthogonal projection onto a subset of variables (that contains all singletons) is integral.
\end{remark}

The next two subsections are dedicated to the proof of \cref{thm:simpleLinearizationProjectionIntegral}.

\subsection{Sufficiency}
\label{sec:integralitySufficiency}

In this section, we prove one direction of \cref{thm:simpleLinearizationProjectionIntegral} by showing that Properties~\ref{PropertyCycle} and~\ref{PropertyPathAbove} are sufficient for integrality of the corresponding projection.
We first prove an auxiliary result, showing that a combination of certain TDI systems results in a TDI system.
Since \AND-relaxations of single monomials are TDI by \cref{thm:ANDisTDI}, this result will allow us to deduce integrality of~$P(\linearization)$, provided~$G(D(\linearization))$ is essentially acyclic.

\begin{proposition}
   \label{thm:glueBinaryTDI}
   For $i \in [2]$, let $A^ix^i + b^i y \leq d^i$ be
   totally dual integral inequality systems that each imply $0 \leq y \leq 1$,
   where
   $A^i \in \Z^{m_i \times n_i}$ and $b^i$, $d^i \in \Z^{m_i}$.
   Then also the combined inequality system
   \begin{alignat*}{8}
      A^1 x^1  &  &         &\,+\,& b^1 y &\leq\,d^1 \\
               &\,& A^2 x^2 &\,+\,& b^2 y &\leq\,d^2
   \end{alignat*}
   is TDI.
\end{proposition}

\begin{proof}
   Consider an objective vector $(c^1, c^2, \delta) \in \Z^{n_1} \times \Z^{n_2} \times \Z$ for which the combined inequality system attains an optimal solution.
   Let, for $i \in [2]$ and $k \in \B{}$,
   \begin{gather*}
      z^i_k \define \sup \setdef{ \transpose{(c^i)} x }[ A^i x + b^i y \leq d^i \text{ and } y = k ] \in \R \cup \setdef{\pm \infty}.
   \end{gather*}
   Due to $0 \leq y \leq 1$, the feasible region of each of these four LPs is a face of an integral polyhedron, which implies $z^i_k \in \Z \cup \setdef{\pm \infty}$.
   
   Suppose (at least) one of the values is $-\infty$, say $z^1_0$.
   If also $z^1_1 = -\infty$, the first system is infeasible and thus the same holds for the combined system.
   Otherwise, the first system implies $y = 1$.
   We can now project $y$ out of the first system and consider the face defined by $y = 1$ for the second system.
   This maintains TDIness as well as the feasible region of the combined system.
   The resulting polyhedron is a Cartesian product of two polyhedra that are defined by TDI inequality systems, proving that the combined system is TDI.

   Consequently, we can assume that $z^i_k \neq -\infty$ for all $i$ and $k$ in the following.
   If one of the values, say $z^1_k$, is $\infty$, then there exists an unbounded ray $(r,p) \in \R^{n_1} \times \R$ with $\transpose{(c^1)} r + \delta p > 0$.
   From $0 \leq y \leq 1$, we obtain $p = 0$. This shows that $(r,\zerovec,0)$ is an unbounded ray in the combined system and satisfies $\transpose{(c^1)} r + \transpose{(c^2)} \zerovec + \delta 0 > 0$.
   This contradicts our assumption that the objective $(c^1, c^2, \delta)$ is bounded over the combined system.

   Otherwise, all four values are finite. Then, for $\delta_1$, $\delta_2 \in
   \R$, consider the finite values
   \begin{equation}\label{eq:TDIsubmax}
     \bar{z}^i \define \max\big\{\transpose{(c^i)} x + \delta_i \cdot y \suchthat A^i x + b^i y \leq d^i\big\}.
   \end{equation}
   Note that if $\delta_i \geq z^i_0 - z^i_1$, then $\bar{z}^i$
   is attained at a vector with $y=1$, because we have $\bar{z}^i \geq
   \max\{z^i_0, z^i_1 + \delta_i\} \geq z^i_0$. Similarly, if
   $\delta_i \leq z^i_0 - z^i_1$, then $\bar{z}^i$ is attained
   at a vector with $y=0$.

   If $z^1_0 - z^1_1 + z^2_0 - z^2_1 \geq \delta$ holds, then choose $\delta_1 \define z^1_0 - z^1_1$ and $\delta_2
   \define \delta - \delta_1 \leq z^2_0 - z^2_1$.
   Hence, $\bar{z}^i$ is attained by $(\bar{x}^i, \bar{y})$ with
   $\bar{x}^i \in \Z^{n_i}$ and $\bar{y} = 0$ for both $i \in [2]$.

   Otherwise, i.e., if $z^1_0 - z^1_1 + z^2_0 - z^2_1 <
   \delta$ holds, then choose $\delta_1 \define z^1_0 - z^1_1$ and $\delta_2 \define \delta - \delta_1 \geq z^2_0 - z^2_1$. 
   Then $\bar{z}^i$ is attained by $(\bar{x}^i, \bar{y})$ with
   $\bar{x}^i \in \Z^{n_i}$ and $\bar{y} = 1$ for both $i \in [2]$.

   In both cases, this shows that $(\bar{x}^1, \bar{x}^2, \bar{y})$ is a maximum for the combined system with respect to $(c^1, c^2, \delta)$.
   Hence, by integrality of the data and total dual integrality of each of the systems, there exist dual multipliers $\lambda^i \in \Z_+^{m_i}$ with
   \[
   \transpose{(\lambda^i)} A^i = \transpose{(c^i)},\quad
   \transpose{(\lambda^i)} b^i = \delta_i,\quad
   \transpose{(\lambda^i)} d^i = \transpose{(c^i)}
   \bar{x}^i + \delta_i \cdot \bar{y}.
   \]
   Thus, $\transpose{(\lambda^1)} d^1 + \transpose{(\lambda^2)} d^2 = \transpose{(c^1)} \bar{x}^1 + \transpose{(c^2)} \bar{x}^2 + \delta \cdot \bar{y}$, i.e., $(\lambda^1, \lambda^2)$ is an integral optimum for the dual of the LP over the combined system.
   This concludes the proof.
\end{proof}

\begin{lemma}
   \label{thm:acyclicTDI}
   Let~$\linearization = (n, \allMonomials, \constraints)$ be a simple linearization.
   If $G(D(\linearization))$ is acyclic, then system~\eqref{eq:linearizationBounds}--\eqref{eq:linearizationSum} is TDI.
\end{lemma}

\begin{proof}
   We prove the statement by induction on the number $\card{\constraints}$ of constraints.
   For $\constraints = \varnothing$, consistency of $\linearization$ implies that $\allMonomials = \singletonMonomials$, and hence $P(\linearization) = [0,1]^n$.
   
   If $G(D(\linearization))$ has more than one connected component, $P(\linearization)$ is the Cartesian product of polytopes of smaller linearizations, whose inequality systems~\eqref{eq:linearizationBounds}--\eqref{eq:linearizationSum} are TDI as well as the combined inequality system.
   Thus, we assume that $G(D(\linearization))$ is connected.

   If $\constraints \neq \varnothing$, we choose a constraint $\constraint^\star \in \constraints$ such that its corresponding monomial $m^\star \define \bigcup \constraint^\star$ has in-degree~$0$.
   We consider the linearization $\linearization' \define (n, \allMonomials', \constraints')$ with $\allMonomials' \define \allMonomials \setminus \setdef{m^\star}$ and $\constraints' \define \constraints \setminus \setdef{\constraint^\star}$.
   Clearly, $\linearization'$ satisfies the assumption of the lemma, since $D(\linearization')$ is a subgraph of $D(\linearization)$.
   By connectivity and acyclicity of $G(D(\linearization))$, the graph obtained from $G(D(\linearization))$ by removing node $m^\star$ consists of $\card{\constraint^\star}$-many connected components, each of which contains a node corresponding to some monomial $m \in \constraint^\star$.
   Let $\linearization_m$ be the linearization whose graph $G(D(\linearization_m))$ is the connected component containing node $m \in \constraint^\star$.
   By the induction hypothesis, the inequality system~\eqref{eq:linearizationBounds}--\eqref{eq:linearizationSum} for $\linearization_m$ is TDI.

   Consider the polytope $P \subseteq \R^{\constraint^\star} \times \R$ with variables $y_m \in [0,1]$ for each $m \in \constraint^\star$ as well as $y_{m^\star} \in [0,1]$ defined by inequalities~\eqref{eq:linearizationPair} and~\eqref{eq:linearizationSum} for constraint $\constraint^\star$.
   By \cref{thm:ANDisTDI}, the system is TDI and $P$ is integral.
   Observe that $P$ has exactly one variable in common with the polytopes $P(\linearization_m)$
   for each $m \in \constraint^\star$.
   By repeated application of Proposition~\ref{thm:glueBinaryTDI}, the combined system (which is system~\eqref{eq:linearizationBounds}--\eqref{eq:linearizationSum} for $\linearization$) is TDI.
   This concludes the proof.
\end{proof}

\cref{thm:acyclicTDI} is the main ingredient for the sufficiency direction of \cref{thm:simpleLinearizationProjectionIntegral}.
Property~\ref{PropertyPathAbove} of the theorem is used to replace the considered linearization $\linearization$ by an equivalent one whose graph is acyclic.

\begin{proof}[Sufficiency proof for \cref{thm:simpleLinearizationProjectionIntegral}.]
   Let~$\linearization = (n, \allMonomials, \constraints)$ be a simple linearization and~$\targetMonomials \subseteq \properMonomials$ be a subset of the proper monomials.
   Assume that $D(\linearization)$ does not contain a subgraph~$Z$ satisfying Properties~\ref{PropertyCycle} and~\ref{PropertyPathAbove}.
   We have to show that $\proj[\singletonMonomials \cup \targetMonomials][P(\linearization)]$ is integral, and, if in addition $\properMonomials \subseteq \succ(\targetMonomials)$ holds, that system~\eqref{eq:linearizationBounds}--\eqref{eq:linearizationSum} is TDI.

   To this end, consider the linearization $\linearization' = (n, \allMonomials', \constraints')$ with $\allMonomials' \define \singletonMonomials \cup \succ(\targetMonomials)$ defined in \cref{thm:redundantMonomials}.
   Note that the construction is such that $\allMonomials'$ contains all target monomials~$\targetMonomials$.
   The corollary states that $P(\linearization') = \proj[\allMonomials'][P(\linearization)]$ holds.
   Moreover, our assumption that $D(\linearization)$ does not contain a subgraph $Z$ satisfying Properties~\ref{PropertyCycle} and~\ref{PropertyPathAbove} implies that $G(D(\linearization'))$ is acyclic.
   Thus, by \cref{thm:acyclicTDI}, system~\eqref{eq:linearizationBounds}--\eqref{eq:linearizationSum} for $\linearization'$ is TDI.
   This already proves the second statement, since $\allMonomials' = \allMonomials$ holds in case $\properMonomials \subseteq \succ(\targetMonomials)$ holds.

   Since the right-hand side of the system is integer, $P(\linearization')$ is integral, which implies integrality of $\proj[\singletonMonomials \cup \targetMonomials][P(\linearization')]$.
   We obtain
   \begin{align*}
      \proj[\singletonMonomials \cup \targetMonomials][P(\linearization)]
      &= \proj[\singletonMonomials \cup \targetMonomials][\proj[\allMonomials'][P(\linearization)]]
      = \proj[\singletonMonomials \cup \targetMonomials][P(\linearization')],
   \end{align*}
   where the second equation holds by \cref{thm:redundantMonomials}, noting that $\singletonMonomials \cup \targetMonomials \subseteq \allMonomials'$ holds.
   This establishes integrality of $\proj[\singletonMonomials \cup \targetMonomials][P(\linearization)]$ and concludes the proof.
\end{proof}

\begin{remark}
   The integrality implications of \cref{thm:ANDisTDI} and of \cref{thm:glueBinaryTDI} are well known (see~\cite{ErmelW18} for a short direct proof and~\cite{ConfortiP16,Margot95} for the more general concept of the ``projected faces property'').
   Hence, the sufficiency proof for the integrality property essentially reduces to \cref{thm:redundantMonomials} and \cref{thm:acyclicTDI} if one exploits this knowledge.
   However, we established the TDI property of the corresponding inequality systems, which is a stronger result.
\end{remark}

\subsection{Necessity}
\label{sec:integralityNecessity}

We now present the other direction of the proof of \cref{thm:simpleLinearizationProjectionIntegral}.
To this end we show that the presence of subgraphs $Z$ of $D(\linearization)$ fulfilling Properties~\ref{PropertyCycle} and~\ref{PropertyPathAbove} in \cref{thm:simpleLinearizationProjectionIntegral} implies fractionality of~$P(\linearization)$. We will consider
subgraphs of~$D(\linearization)$ and~$G(D(\linearization))$.
If we are given a path in~$D(\linearization)$ or~$G(D(\linearization))$, we always assume that it is given by its arcs or edges, respectively.

\begin{figure}[tb]
   \centering
   \begin{tikzpicture}
      \node[nodeMonomial,label={left:$u_1$}] at (0,4) (u1) {};
      \node[nodeMonomial] at (-1,3) (v1) {};
      \node[nodeMonomial,label={left:$u_2$}] at (2,4) (u2) {};
      \node[nodeMonomial] at (1,3) (v2) {};
      \node[nodeMonomial,label={left:$m$}] at (3,3) (v3) {};
      \node[nodeMonomial] at (4,2) (v4) {};
      \node[nodeMonomial,label={left:$u_3$}] at (5,4) (u3) {};
      \node[nodeMonomial] at (6,3) (v5) {};
      \node[nodeMonomial,label={below:$m'$}] at (3,1) (v6) {};

      \node at (2,2) {$Z_1^\star$};
      \node at (5,2) {$Z_2^\star$};

      \draw[path] (u1) -> (v1);
      \draw[path] (u1) -> (v2);
      \draw[path] (u2) -> (v2);
      \draw[path] (u2) -> (v3);
      \draw[path] (v3) -> (v4);
      \draw[path] (u3) -> (v4);
      \draw[path] (u3) -> (v5);
      \draw[path] (v1) .. controls +(down:2) and +(left:2) .. (v6);
      \draw[path] (v5) .. controls +(down:2) and +(right:2) .. (v6);

      \draw[path] (v3) -> (v6) node[pos=0.5,right]{$P$};
   \end{tikzpicture}
   \caption{%
      Illustration of a shortcut path in the proof of \cref{thm:SimpleLinearizationCyclicExtraProperties}. Here, $\upperNodes(Z_1^{\star}) = \upperNodes(Z_1^{\star} \cup P) = \setdef{u_1,u_2}$, $\upperNodes(Z_2^{\star}) = \setdef{u_3}$ and $\upperNodes(Z_2^{\star} \cup P) = \setdef{u_3,m}$ hold.}
   \label{fig:shortcut}
\end{figure}

For subgraphs $Z$ of $D$, let~$\upperNodes(Z) \subseteq V(Z)$ be the set of nodes with out-degree at least~$2$ in the subgraph $Z$, called \emph{upper nodes}, and let~$\lowerNodes(Z) \subseteq V(Z)$ be the nodes with in-degree at least~$2$ in the subgraph~$Z$, called \emph{lower nodes}.
Note that if $G(Z)$ is a cycle, then \mbox{$\card{\upperNodes(Z)} = \card{\lowerNodes(Z)}$} holds.
Moreover, in this case the requirement ``$V(Z) \subseteq \succ(\targetMonomials)$'' in \cref{thm:simpleLinearizationProjectionIntegral} is equivalent to ``$\upperNodes(Z) \subseteq \succ(\targetMonomials)$'' since $V(Z) \subseteq \succ(\upperNodes(Z))$ holds.
In \cref{fig:simpleLinearizationProjectionIntegralEvil}, $\upperNodes(Z) = \{ \{1,2\}, \{1,3\}, \{2,3,4\} \}$ and $\lowerNodes(Z) = \{ \{1\}, \{2\}, \{3\} \}$ hold, while in \cref{fig:simpleLinearizationProjectionIntegralNice}, we have $\upperNodes(Z) = \{ \{2,3,4\} \}$ and $\lowerNodes(Z) = \{ \{3\} \}$.

\begin{lemma}
   \label{thm:SimpleLinearizationCyclicExtraProperties}
   Let $\linearization = (n, \allMonomials, \constraints)$ be a simple linearization, and let $\targetMonomials \subseteq \properMonomials$ be a subset of the proper monomials.
   If $D(\linearization)$ has a subgraph that satisfies Properties~\ref{PropertyCycle} and~\ref{PropertyPathAbove} of \cref{thm:simpleLinearizationProjectionIntegral}, then $D(\linearization)$ has a subgraph $Z$ satisfying Properties~\ref{PropertyCycle} and~\ref{PropertyPathAbove} and $\card{\upperNodes(Z)} = 1$ or the following three properties:
   \begin{enumerate}[label=(\alph{*}),start=3]
   \item
      \label{PropertyUpperLowerPath}
      for all $s \in \singletonMonomials$, $t \in \targetMonomials$, there is at most one $t$-$s$-path,
   \item\label{PropertyLowerEmptyIntersection} $\card{\pred(s) \cap \lowerNodes(Z)} \leq 1$ for all $s \in \singletonMonomials$,
   \item\label{PropertyUpperEmptyIntersection} $\card{\succ(t) \cap \upperNodes(Z)} \leq 1$ for all $t \in \targetMonomials$.
   \end{enumerate}
\end{lemma}

\begin{proof}
   Let $Z^\star$ be a subgraph satisfying Properties~\ref{PropertyCycle}
   and~\ref{PropertyPathAbove} with minimal $\card{\upperNodes(Z^\star)}$.
   Since $D(\linearization)$ is acyclic and $G(Z^\star)$ is a cycle, we have $\card{\upperNodes(Z^\star)} \geq 1$.
   If $\card{\upperNodes(Z^\star)} = 1$, we are done, so we assume
   from now on that $\card{\upperNodes(Z^\star)} \geq 2$ holds.
   We claim that $Z^\star$ also satisfies Properties~\ref{PropertyUpperLowerPath}--\ref{PropertyUpperEmptyIntersection}.

   To prove Property~\ref{PropertyUpperLowerPath}, assume that for some $t \in \targetMonomials$ and some $s \in \singletonMonomials$, there are two distinct $t$-$s$-paths $P_1$ and $P_2$ in $D(\linearization)$.
   Then the symmetric difference of $P_1$ and $P_2$ contains a subgraph $Z'$ that consists of
   two internally node-disjoint paths having the same start node and the same terminal node.
   Clearly, $Z'$ satisfies Property~\ref{PropertyCycle} and~\ref{PropertyPathAbove}, since $t$ is a predecessor of all nodes in $V(Z') \supseteq \upperNodes(Z')$.
   The construction implies $\card{\upperNodes(Z')} = 1$, contradicting our assumption.

   Before proving Property~\ref{PropertyLowerEmptyIntersection} and~\ref{PropertyUpperEmptyIntersection}, we show that there is no directed path in $D(\linearization)$ between two distinct nodes of $V(Z^\star)$ that is arc-disjoint with $Z^\star$, see \cref{fig:shortcut} for an illustration.
   We will subsequently call such a path a \emph{shortcut path}.
   Assume for the sake of contradiction that $P$ is such a path from node $m$ to $m'$.
   The two nodes $m$ and $m'$ induce a partition of $Z^\star$ into two distinct subgraphs $Z_1^\star$ and $Z_2^\star$ such that $G(Z_1^\star)$ and $G(Z_2^\star)$ are paths in $G(Z^\star)$ connecting $m$ and $m'$.
   Then, for~$i \in [2]$, $G(Z_i^\star \cup P)$ is a cycle.
   Since $P$ is a directed path, it can only induce one more node (namely $m$) into $\upperNodes(Z^\star)$, i.e. $\card{\upperNodes(Z_i^\star \cup P)} \leq \card{\upperNodes(Z_i^\star)} + 1$, but only one case can attain equality.
   Moreover, one $Z_i^\star$ contains at most half of the nodes of $\upperNodes(Z^\star)$.
   Since $\card{\upperNodes(Z^\star)} \geq 2$ it follows that $\card{\upperNodes(Z_i^\star \cup P)} < \card{\upperNodes(Z^\star)}$ for one $i \in [2]$.
   Moreover, when following~$m$ backwards in $Z^\star$, we reach a node $u \in \upperNodes(Z^\star)$.
   By assumption on $Z^\star$, $\upperNodes(Z^\star) \subseteq \succ(\targetMonomials)$.
   Thus, the same holds for $Z_i^\star \cup P$, i.e., Property~\ref{PropertyPathAbove} holds for $Z_i^\star \cup P$.
   This contradicts our minimality assumption on the choice of $Z^\star$.
   
   To prove Property~\ref{PropertyLowerEmptyIntersection}, assume that there exists a singleton $s \in \singletonMonomials$
   and two distinct monomials $\ell_1$, $\ell_2 \in \pred(s) \cap \lowerNodes(Z^\star)$ with
   corresponding $\ell_1$-$s$-path $P_1$ and $\ell_2$-$s$-path $P_2$ in~$D(\linearization)$, respectively.
   There exists a node $m \in V(P_1) \cap V(P_2)$ such that $P_1$ and $P_2$ contain arc-disjoint
   $\ell_1$-$m$- and $\ell_2$-$m$-paths $P_1'$ and $P_2'$, respectively.
   Moreover, $P'_1$ and~$P_2'$ are arc-disjoint from $Z^\star$, since they would otherwise induce a shortcut path.
   The two nodes $\ell_1$, $\ell_2 \in V(Z^\star)$ induce a partition of $Z^\star$
   into two distinct subgraphs $Z_1^\star$ and $Z_2^\star$ such that $G(Z_1^\star)$ and $G(Z_2^\star)$ are
   paths in $G(Z^\star)$ connecting $\ell_1$ and $\ell_2$.
   Clearly, $\upperNodes(Z_1^\star) \cap \upperNodes(Z_2^\star) = \varnothing$ holds. We assume w.l.o.g.\ $\upperNodes(Z_1^\star) \subset \upperNodes(Z^\star)$.
   Consider the subgraph $Z' \define P'_1 \cup P'_2 \cup Z_1^\star$, and observe that $G(Z')$ is a cycle, since $P'_1$, $P'_2$ and $Z_1^\star$ are disjoint by construction.
   Since $\ell_1$, $\ell_2 \in \lowerNodes(Z^\star)$, $Z'$ does not add nodes to $\upperNodes(Z_1^\star)$. Thus, $\upperNodes(Z') = \upperNodes(Z_1^{\star})\subset \upperNodes(Z^\star)$ holds, which implies Property~\ref{PropertyPathAbove} for $Z'$.
   Then $\card{\upperNodes(Z')} < \card{\upperNodes(Z^\star)}$ contradicts our minimality assumption on $\card{\upperNodes(Z^\star)}$ and establishes Property~\ref{PropertyLowerEmptyIntersection} for $Z^\star$.

   The proof for Property~\ref{PropertyUpperEmptyIntersection} is similar.
   Assume that there exists a target monomial $t \in \targetMonomials$ and two distinct monomials $u_1$, $u_2 \in \succ(t) \cap \upperNodes(Z^\star)$ with
   corresponding $t$-$u_1$-path $P_1$ and $t$-$u_2$-path $P_2$ in $D(\linearization)$.
   There exists a node $m \in V(P_1) \cap V(P_2)$ such that $P_1$ and $P_2$ contain arc-disjoint
   $m$-$u_1$- and $m$-$u_2$-paths $P_1'$ and $P_2'$, respectively.
   Moreover, $P_1'$ and~$P_2'$ are arc-disjoint from $Z^\star$, since they would otherwise induce a shortcut path.
   The two nodes $u_1$, $u_2 \in V(Z^\star)$ induce a partition of $Z^\star$ into two
   disjoint subgraphs $Z_1^\star$ and $Z_2^\star$ such that $G(Z_1^\star)$ and $G(Z_2^\star)$ are paths in
   $G(Z^\star)$ connecting $u_1$ and $u_2$.
   Observe that $G(Z')$ is a cycle for the subgraph $Z' \define Z^\star_1 \cup P'_1 \cup P'_2$, since $P'_1$, $P'_2$ and~$Z^\star_1$ are disjoint by construction.
   Since~$m$ is a predecessor of both~$u_1$ and~$u_2$ in~$D(\linearization)$, neither~$u_1$ nor~$u_2$ can be in~$\upperNodes(Z')$.
   Thus, $\upperNodes(Z') \subseteq (\upperNodes(Z^\star) \cup \setdef{m}) \setminus \setdef{u_1, u_2}$ holds,
   which implies Property~\ref{PropertyPathAbove} for $Z'$, since $t \in \pred(m) \cap \targetMonomials$.
   It also implies $\card{\upperNodes(Z')} < \card{\upperNodes(Z^\star)}$, which contradicts our minimality assumption on $\card{\upperNodes(Z^\star)}$ and establishes Property~\ref{PropertyUpperEmptyIntersection} for~$Z^\star$.
\end{proof}

\begin{figure}[tb]
   \begin{tikzpicture}
      \node[nodeSingleton,label=right:1] (1) at (0,0) {};
      \node[nodeSingleton,label=right:1] (2) at (2,0) {};
      \node[nodeSingleton,label=right:$\frac{2}{3}$] (3) at (4,0) {};
      \node[nodeSingleton,label=left:1] (4) at (6,0) {};
      \node[nodeSingleton,label=right:1] (6) at (10,0) {};
      \node[nodeSingleton,label=right:1] (5) at (8,0) {};

      \node[nodeMonomial,label=left:1] (12) at (1,1) {};
      \node[nodeMonomial,label=left:$\frac{2}{3}$] (23) at (3,1) {};
      \node[nodeMonomial,label=left:$\frac{2}{3}$] (34) at (5,1) {};
      \node[nodeMonomial,label=right:1] (46) at (9,1) {};

      \node[nodeMonomial,label=left:$\frac{2}{3}$] (123) at (2,2) {};
      \node[nodeMonomial,label=left:$\frac{1}{3}$] (234) at (4,2) {};
      \node[nodeTarget,label=left:$\frac{2}{3}$,fill=black] (345) at (6,2) {};
      \node[nodeTarget,label=right:1] (456) at (8.5,2) {};

      \node[nodeTarget,label=left:0] (1234) at (3,3) {};

      \node[nodeMonomial,label=left:0] (12345) at (4,4) {};

      \foreach \a/\b in {%
         12345/1234, 12345/345,
         123/12,
         234/23,
         345/34, 345/5,
         456/46, 456/5,
         12/1, 12/2,
         23/2,
         34/4,
         46/4,
         46/6%
         }%
      {
         \draw[arc] (\a) -- (\b);
      }

      \foreach \a/\b in {%
         1234/123, 1234/234,
         123/23,
         234/34,
         23/3,
         34/3%
         }%
      {
         \draw[arc, dashed, ultra thick] (\a) -- (\b);
      }
   \end{tikzpicture}
   \caption{The point~$y$ constructed in the proof of \cref{thm:simpleLinearizationCyclicCardinalityOne} w.r.t.\ the dashed cycle in the linearization of \cref{ex:runningExample}.}
   \label{fig:runningExampleKmConstruction}
\end{figure}
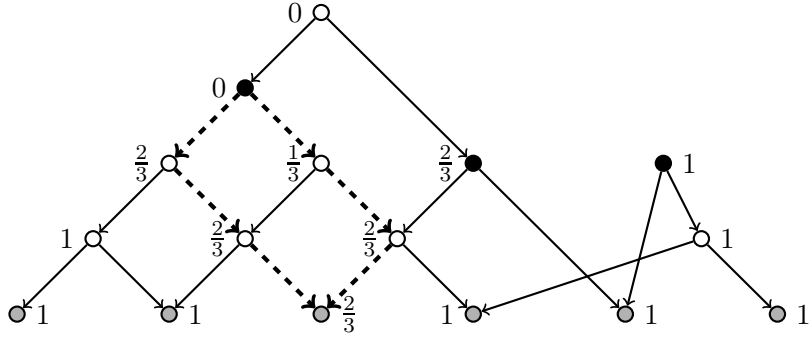

\begin{lemma}
   \label{thm:simpleLinearizationCyclicCardinalityOne}
   Let $\linearization = (n, \allMonomials, \constraints)$ be a simple linearization and let $\targetMonomials \subseteq \properMonomials$ be a subset of the proper monomials.
   If $D(\linearization)$ has a subgraph $Z$ that satisfies~\ref{PropertyCycle}, \ref{PropertyPathAbove} and $\card{\upperNodes(Z)} = 1$, then $\proj[\singletonMonomials \cup \targetMonomials][P(\linearization)]$ is not integral.
\end{lemma}

\begin{proof}
   Let $Z$ be as stated in the lemma, and let $t \in \targetMonomials$ be such that there exists a $t$-$u$-path in $D(\linearization)$, where $u \in \upperNodes(Z)$.
   Let $\ell \in \lowerNodes(Z)$ and $s \in \singletonMonomials \cap \succ(\ell)$.
   Note that since $\card{\upperNodes(Z)} = \card{\lowerNodes(Z)} = 1$ holds, $u$ and $\ell$ are unique.
   For $m \in \allMonomials$, we denote by $k(m) \in \Z_+$ the number of distinct $m$-$s$-paths in $D(\linearization)$.
   Clearly, $Z$ consists of two directed $u$-$\ell$-paths that can be extended to $u$-$s$-paths in $D(\linearization)$, which proves $k(u) \geq 2$.
   Note that all $\constraint \in \constraints$ satisfy
   \begin{equation}
      k\big(\bigcup \constraint\big) = \sum_{m \in \constraint} k(m), \label{EquationNumPathCount}
   \end{equation}
   since the $\bigcup \constraint$-$s$ paths are exactly those paths that start with some arc $(\bigcup \constraint, m)$ (with $m \in \constraint$) and continue with some $m$-$s$-path.
   Moreover, we have $k(\hat{s}) = 0$ for each $\hat{s} \in \singletonMonomials \setminus \{s\}$ as well as $k(s) = 1$.
   We define the point $y \in \R^{\allMonomials}$ via 
   \begin{align*}
      y_m & \define \max \setdef{0,\; 1 - \tfrac{k(m)}{k(u)} }\quad \text{ for each monomial } m \in \allMonomials.
   \end{align*}
   \Cref{fig:runningExampleKmConstruction} shows the constructed point~$y$
   for the running \cref{ex:runningExample}.
   Observe that $y_{\hat{s}} = 1$ for each $\hat{s} \in \singletonMonomials \setminus \{s\}$, since $k(\hat{s}) = 0$.
   Moreover, $y_{s} = 1 - 1/k(u) \in (0,1)$ holds.
   Finally, since the $k(u)$-many distinct $u$-$s$-paths induce at least $k(u)$-many $t$-$s$-paths, we have $k(t) \geq k(u)$.
   Thus, $y_t$ is the maximum of $0$ and $1 - k(t) / k(u) \leq 1 - k(u) / k(u) = 0$.
   Hence, $y_t = 0$.

   We now prove that $y \in P(\linearization)$.
   By construction, $y \in [0,1]^{\allMonomials}$.
   Consider a constraint $\constraint \in \constraints$, and let $m' \define \bigcup \constraint$.
   For each $m \in \constraint$, we have $k(m') \geq k(m)$ since $(m',m)$ is an arc in $D(\linearization)$.
   If $y_m = 0$, we have $k(m') \geq k(m) \geq k(u)$, which implies $y_{m'} = 0$.
   Otherwise, if $y_{m'} > 0$ then $y_{m'} = 1 - k(m') / k(u) \leq 1 - k(m) / k(u) = y_m$.
   We conclude that~\eqref{eq:linearizationPair} holds for each $m \in \constraint$.
   If $y_m = 0$ for some $m \in \constraint$, then~\eqref{eq:linearizationSum} is clearly satisfied.
   Otherwise we have
   \begin{align*}
      \sum_{m \in \constraint} y_m
      &= \card{\constraint} - \sum_{m \in \constraint} \frac{k(m)}{k(u)}
      \underset{\eqref{EquationNumPathCount}}{=} \card{\constraint} - \frac{k(m')}{k(u)}
      = 1 - \frac{k(m')}{k(u)} + \card{\constraint} - 1
      = y_{m'} + \card{\constraint} - 1.
   \end{align*}
   Thus, \eqref{eq:linearizationSum} is satisfied.

   Consider the projection  $y' \in \proj[\singletonMonomials \cup \targetMonomials][y]$ of $y$ and assume that $y'$ lies in the polytope $Q \define \conv(\proj[\singletonMonomials \cup \targetMonomials][P(\linearization)] \cap \Z^{\singletonMonomials \cup \targetMonomials})$.
   Then $y'$ is the convex combination of integer points $y^{(1)}$, $y^{(2)}$, \dots, $y^{(q)} \in Q \cap \Z^{\singletonMonomials \cup \targetMonomials}$.
   By construction of~$y$, the integer points have to satisfy $\smash{y^{(k)}_t} = 0$ and $\smash{y^{(k)}_{\hat{s}}} = 1$ for all $\hat{s} \in \singletonMonomials \setminus \{s\}$ for all $k \in [q]$.
   This implies $\smash{y^{(k)}_s = 0}$ for every~$k \in [q]$, because
   otherwise, $\smash{y^{(k)}_t}$ has to be~1. Thus, $y'_s = 0$, a contradiction.
   This proves that the fractional point $y'$ does not lie in $Q$, i.e., $\proj[\singletonMonomials \cup \targetMonomials][P(\linearization)]$ is not integral.
\end{proof}

\begin{figure}[tb]
   \centering
   \begin{tikzpicture}
      \node[nodeSingleton,label=right:1] (1) at (0,0) {};
      \node[nodeSingleton,label=right:1] (2) at (2,0) {};
      \node[nodeSingleton,label=right:1] (3) at (4,0) {};
      \node[nodeSingleton,label=left:$\frac{1}{2}$] (4) at (6,0) {};
      \node[nodeSingleton,label=right:1] (6) at (10,0) {};
      \node[nodeSingleton,label=right:$\frac{1}{2}$] (5) at (8,0) {};

      \node[nodeMonomial,label=left:1] (12) at (1,1) {};
      \node[nodeMonomial,label=left:1] (23) at (3,1) {};
      \node[nodeMonomial,label=left:$\frac{1}{2}$] (34) at (5,1) {};
      \node[nodeMonomial,label=right:$\frac{1}{2}$] (46) at (9,1) {};

      \node[nodeMonomial,label=left:1] (123) at (2,2) {};
      \node[nodeMonomial,label=left:$\frac{1}{2}$] (234) at (4,2) {};
      \node[nodeTarget,label=left:0,label=right:$u_2$] (345) at (6,2) {};
      \node[nodeTarget,label=left:$\frac{1}{2}$,label=right:$u_1$] (456) at (8.5,2) {};

      \node[nodeTarget,label=left:$\frac{1}{2}$] (1234) at (3,3) {};

      \node[nodeMonomial,label=left:0] (12345) at (4,4) {};

      \foreach \a/\b in {%
         12345/1234, 12345/345,
         1234/123, 1234/234,
         123/12, 123/23,
         234/23, 234/34,
         12/1, 12/2,
         23/2, 23/3,
         34/3,
         46/6%
         }%
      {
         \draw[arc] (\a) -- (\b);
      }

      \foreach \a/\b in {%
         345/34, 345/5,
         456/46, 456/5,
         34/4,
         46/4%
         }%
      {
         \draw[arc, dashed, ultra thick] (\a) -- (\b);
      }
   \end{tikzpicture}
   \caption{%
      The point~$y$ constructed in the proof of \cref{thm:SimpleLinearizationCyclicCardinalityHigher} w.r.t.\ the dashed cycle in the linearization of \cref{ex:runningExample}.}
   \label{fig:runningExampleKmConstruction2}
\end{figure}

\begin{lemma}
   \label{thm:SimpleLinearizationCyclicCardinalityHigher}
   Let $\linearization = (n, \allMonomials, \constraints)$ be a simple linearization and let $\targetMonomials \subseteq \properMonomials$ be a subset of the proper monomials.
   If $D(\linearization)$ has a subgraph $Z$ that satisfies Properties~\ref{PropertyCycle}--\ref{PropertyUpperEmptyIntersection} and $\card{\upperNodes(Z)} > 1$, then $\proj[\singletonMonomials \cup \targetMonomials][P(\linearization)]$ is not integral.
\end{lemma}

\begin{proof}
   Let $Z$ be a subgraph of $D$ as stated in the lemma.
   Let $\upperNodes(Z) = \setdef{u_1, u_2, \dotsc, u_k}$ and $\lowerNodes(Z) = \setdef{\ell_1, \ell_2, \dotsc, \ell_k}$ with $k > 1$.
   We assume the these nodes are ordered such that $Z$ contains a path from $u_i$ to $\ell_i$ for $i \in [k]$, paths from $u_i$ to $\ell_{i+1}$ for $i \in [k-1]$ as well as a path from $u_k$ to $\ell_1$.
   For $i \in [k]$, choose $s_i \in \singletonMonomials \cap \succ(\ell_i)$, which exists by \cref{thm:consistencyPaths}.
   Similarly, for $i \in [k]$, choose $t_i \in \targetMonomials \cap \pred(u_i)$, which is guaranteed to exist by Property~\ref{PropertyPathAbove}.
   Due to $\targetMonomials \cap \singletonMonomials = \varnothing$ and since $Z$ satisfies Properties~\ref{PropertyLowerEmptyIntersection} and~\ref{PropertyUpperEmptyIntersection}, the monomials $s_i$ and $t_i$ are distinct for all $i \in [k]$.

   We define the point $y \in \R^{\allMonomials}$ via
   \begin{gather*}
      y_m \define \begin{cases} 
         0 & \text{ if } m \in \pred(u_k),\\
         \tfrac{1}{2} & \text{ if } m \in \pred(s_i) \text{ for some $i \in [k]$ and } m \notin \pred(u_k), \\
         1 & \text{ otherwise},
       \end{cases}
     \end{gather*}
   for all $m \in \allMonomials$.
   \Cref{fig:runningExampleKmConstruction2} shows the constructed point for
   the running \cref{ex:runningExample}.
   We first show $y \in P(\linearization)$.
   By construction, the bounds~\eqref{eq:linearizationBounds} are satisfied.

   Consider an inequality~\eqref{eq:linearizationPair} $y_{m'} \leq y_m$ for $m \in \constraint \in \constraints$ and let $m' \define \bigcup \constraint$.
   Thus, $(m', m)$ is an arc of $D(\linearization)$.
   If $y_m = 0$, then $m \in \pred(u_k)$ and also $m' \in \pred(u_k)$, which implies $y_{m'} = 0$, and the inequality holds.
   If $y_m = \tfrac{1}{2}$, then $m \in \pred(s_i)$ for some $i \in [k]$, which implies $m' \in \pred(s_i)$.
   Thus, $y_{m'} = 0$ if $m \in \pred(u_k)$ or $y_{m'} = \tfrac{1}{2}$ otherwise. Again the inequality holds, which is
   also trivially true for $y_m = 1$.
   This shows that~\eqref{eq:linearizationPair} is satisfied.

   Consider an inequality~\eqref{eq:linearizationSum} $\sum_{m \in \constraint} y_m \leq y_{m'} + \card{\constraint} - 1$ for $\constraint \in \constraints$
   with $m' \define \bigcup \constraint$.
   If $y_{m'} = \tfrac{1}{2}$, then $m' \in \pred(s_i)$ for some $i \in [k]$.
   Hence, also $m \in \pred(s_i)$ for some monomial $m \in \constraint$ which shows $y_m = \tfrac{1}{2}$ or $y_m = 0$.
   Thus, the inequality is satisfied.
   If $y_{m'} = 0$, then $m' \in \pred(u_k)$.
   If, for some $m \in \constraint$, $m \in \pred(u_k)$ holds, we have $y_m = 0$ as well.
   Otherwise, $m' = u_k$ holds, and there exist monomials $m_1$, $m_2 \in \constraint$ with $m_1 \in \pred(\ell_k) \subseteq \pred(s_k)$ and $m_2 \in \pred(\ell_1) \subseteq \pred(s_1)$.
   Hence, $y_{m_1} = y_{m_2} = \tfrac{1}{2}$ in this case.
   Thus, in both cases, the inequality holds.
   The only remaining case is $y_{m'} = 1$, in which the inequality 
   follows from $y_m \leq 1$ for all $m \in \constraint$.
   This proves that~\eqref{eq:linearizationSum} holds, and we conclude that $y \in P(\linearization)$.

   We now show $y_{t_i} = \tfrac{1}{2}$ for $i \in [k-1]$.
   From $t_i \in \pred(u_i) \subseteq \pred(\ell_i) \subseteq \pred(s_i)$ it follows that $y_{t_i} \neq 1$.
   We also have $y_{t_i} \neq 0$, since otherwise $u_k \in \succ(t_i) \cap \succ(u_i)$ would contradict Property~\ref{PropertyUpperEmptyIntersection}.
   Hence, $y_{t_i} = \tfrac{1}{2}$.

   Let $Q' \define \proj[\singletonMonomials \cup
   \targetMonomials][P(\linearization)]$, and let $y'$ be the projection of $y$.
   Assume, for the sake of contradiction, that $Q'$ is integral.
   Consider the set~$F$ of points~$y \in Q'$ that satisfy
   \begin{itemize}
   \item
      $y_{\hat{s}} = 1$ for all $\hat{s} \in \singletonMonomials \setminus \setdef{ s_i }[ i \in [k] ]$,
   \item
      $y_{t_k} = 0$,
   \item
      $y_{s_i} = y_{t_i}$ and $y_{s_{i+1}} = y_{t_i}$ for $i \in [k-1]$.
   \end{itemize}
   Note that $F$ is a face of $Q'$, since for $i \in [k-1]$, there exist $t_i$-$s_i$-paths and $t_i$-$s_{i+1}$-paths in $D(\linearization)$, showing $s_i$, $s_{i+1} \subseteq \succ(t_i)$, which in turn shows that the inequalities $y_{t_i} \leq y_{s_i}$ and $y_{t_i} \leq y_{s_{i+1}}$ are valid for $Q'$.
   Thus, since $Q'$ is integral, $F$ is integral as well.

   Let $y^\star \in F \cap \Z^{\singletonMonomials \cup \targetMonomials}$ be an arbitrary integer point in $F$.
   By definition of $F$, there exists a value $\beta \in \setdef{0,1}$ such that $y^\star_{s_i} = \beta$ for all $i \in [k]$ and $y^\star_{t_i} = \beta$ for all $i \in [k-1]$.
   Suppose, $\beta = 1$ holds, i.e., $y^\star_{\hat{s}} = 1$ for all $\hat{s} \in \singletonMonomials$ holds.
   Then \eqref{eq:linearizationSum} and consistency
   of~$\linearization$ imply $y^\star_m = 1$ for all $m \in \allMonomials$, in particular for $m = t_k$, which is a contradiction to $y^\star \in F$.
   Hence, $\beta = 0$ must hold.
   Again by \eqref{eq:linearizationBounds}--\eqref{eq:linearizationSum} and consistency of~$\linearization$ the remaining entries of $y^\star$ are uniquely determined.
   This shows that $F$ contains at most one integer point, namely $y^{\star}$ with $\beta = 0$.
   Since also $y' \in F$ holds, and since $y' \notin \Z^{\singletonMonomials \cup \targetMonomials}$, we have a contradiction to the integrality of $F$.
   This in turn shows that $Q'$ is not integral.
\end{proof}

It now remains to combine the previous lemmas for our necessity proof.

\begin{proof}[Necessity proof for \cref{thm:simpleLinearizationProjectionIntegral}.]
   Let~$\linearization = (n, \allMonomials, \constraints)$ be a simple linearization and~$\targetMonomials \subseteq \properMonomials$ be a subset of the proper monomials.
   Assume that $D(\linearization)$ contains a subgraph~$Z$ satisfying Properties~\ref{PropertyCycle} and~\ref{PropertyPathAbove}.
   We have to show that $\proj[\singletonMonomials \cup \targetMonomials][P(\linearization)]$ is not integral.

   Using the assumptions, \cref{thm:SimpleLinearizationCyclicExtraProperties} implies that we can assume $\card{\upperNodes(Z)} = 1$ or that $Z$ satisfies Properties~\ref{PropertyUpperLowerPath}--\ref{PropertyUpperEmptyIntersection}.
   The fact that $\proj[\singletonMonomials \cup \targetMonomials][P(\linearization)]$ is not integral is implied by \cref{thm:simpleLinearizationCyclicCardinalityOne} in the former case and by \cref{thm:SimpleLinearizationCyclicCardinalityHigher} in the latter case.
   This concludes the proof.
\end{proof}

\section{Existence of Integral Simple Linearizations}
\label{sec:existence}

In the previous section, we were able to characterize whether a given simple linearization induces an integral relaxation of the multilinear polytope.
Extending this investigation, a natural question is whether one can characterize whether a given polynomial admits an integral simple linearization or whether every simple linearization leads to a fractional polytope.
In this section, we provide the following complete answer to this question.

\begin{theorem}
   \label{thm:RIP}
   Let $\targetMonomials$ be a set of monomials with singletons $\singletonMonomials$.
   Then there exists a simple linearization $\linearization = (n, \allMonomials, \constraints)$ with $\singletonMonomials \cup \targetMonomials \subseteq \allMonomials$ such that $\proj[\singletonMonomials \cup \targetMonomials][P(\linearization)]$ is integral if and only if none of the following properties holds:
   \begin{enumerate}[label=(\Alph{*})]
   \item
      \label{MIP1}
      There exist~$m_1$, $m_2$, $m_3 \in \targetMonomials$ with~$m_1 \cap m_2 \cap m_3 \neq \varnothing$ such that $m_3 \cap (m_1 \cup m_2)$ is a proper superset of both~$m_3 \cap m_1$ and~$m_3 \cap m_2$.
   \item
      \label{MIP2}
      There exist pairwise different~$m_1, \dots, m_k \in \targetMonomials$, $k \geq 3$, such that for all~$i, j \in [k]$, we have~$m_i \cap m_j \neq \varnothing$ if and only if~$i$ and~$j$ differ by at most~$1$ modulo~$k$.
   \end{enumerate}
\end{theorem}

Note that the monomials~$m_1,m_2,m_3$ in \ref{MIP1} need to be pairwise different,
since otherwise, $m_3 \cup (m_1 \cup m_2)$ cannot be a proper superset of
$m_3 \cap m_i$ for both $i \in [2]$.

In the following, we say that a set of monomials~$\targetMonomials$ has the \emph{monomial intersection property} if it neither satisfies~\ref{MIP1} nor~\ref{MIP2}.
It is intuitive that Property~\ref{MIP2} induces a cycle in $G(D(\linearization))$ for any suitable linearization $\linearization$.
However, seeing that Property~\ref{MIP1} has the same consequence is not as easy, and we start with an example. We use the following notation:
Given two distinct nodes~$u$ and~$v$ of a digraph, we denote by~$P(u,v)$ a path connecting~$u$ and~$v$ (if it exists).
Moreover, if~$v = \setdef{i}$ for some~$i \in [n]$, we let~$P(u,i) \coloneqq P(u,\setdef{i})$ whenever the meaning of~$i$ is clear from the context.

\begin{example}
   \label{ex:tripleMonomial}
   Consider~$\targetMonomials = \setdef{m_1, m_2, m_3}$ with $m_1 = \setdef{1,2,3,4}$, $m_2 = \setdef{4,5,6}$, and $m_3 = \setdef{3,4,5}$ with singletons $\singletonMonomials = \setdef{1,2,3,4,5,6}$.
   The relevant intersections are given by $m_3 \cap m_1 = \setdef{3,4}$, $m_3 \cap m_2 = \setdef{4,5}$, $m_3 \cap (m_1 \cup m_2) = m_3$ and $m_1 \cap m_2 \cap m_3 = \setdef{4}$.
   Since $m_3 \cap (m_1 \cup m_2)$ is a proper superset of both~$m_1 \cap m_3$ and~$m_2 \cap m_3$, $\targetMonomials$ has Property~\ref{MIP1}.
   
   First observe that this property implies that neither~$m_1 \cap m_3 \subseteq m_2 \cap m_3$ nor that the reverse inclusion holds.
   By \cref{thm:consistencyPaths}, there exist paths~$P(m_2,4)$ and~$P(m_3,4)$ connecting~$m_2$ and~$m_3$ with~$\{4\}$ as well as paths~$P(m_2,5)$ and~$P(m_3,5)$ connecting~$m_2$ and~$m_3$ with~$\{5\}$ in~$D(\linearization)$.
   
   Suppose~$m_2 \cap m_3 = \{4,5\}$ is not a successor of~$m_2$.
   Then, at least one of~$P(m_2,4)$ or~$P(m_2,5)$ consists of the single arc~$(m_2, \{4\})$ or~$(m_2,\{5\})$, respectively. Hence, the concatenation of all four paths contains a cycle in~$G(D(\linearization))$, because neither~$m_2$ can be a successor of~$m_3$ nor vice versa, see \cref{fig:tripleMonomialMissing}.
   Thus, such a linearization cannot be integral by \cref{thm:simpleLinearizationProjectionIntegral}.
   Similarly, we can argue if~$m_2 \cap m_3$ is not a successor of~$m_3$ or~$m_1 \cap m_2$ is neither a successor of~$m_1$ nor~$m_2$.

   For this reason, any integral linearization contains~$m_1 \cap m_3$ and~$m_2 \cap m_3$ as monomials and both have to be successors of~$m_1$ and~$m_3$ as well as~$m_2$ and~$m_3$, respectively.
   However, since~$m_3 \cap m_1 \nsubseteq m_3 \cap m_2$ and~$m_3 \cap m_1 \nsupseteq m_3 \cap m_2$, consistency of~$\linearization$ implies that there exist two distinct
  paths connecting~$m_3$ and~$m_1 \cap m_2 \cap m_3$, one
  using~$m_3 \cap m_1$ as intermediate node and the other
  using~$m_3 \cap m_2$. Thus, $G(D(\linearization))$ contains a cycle, see
  \cref{fig:tripleMonomialsPresent},
  showing that no linearization of~$\targetMonomials$ can be integral.
\end{example}

\begin{figure}[tb]
   \begin{subfigure}[b]{0.36\textwidth}
      \centering
      \begin{tikzpicture}
         \node (345) at (0,0) [nodeMonomial,label=above:\footnotesize{ $\{3,4,5\}$}] {};
         \node (456) at (2,0) [nodeMonomial,label=above:\footnotesize{ $\{4,5,6\}$}] {};
         \node (4) at (0.5,-2) [nodeSingleton,label=below:\footnotesize{ $\{4\}$}] {};
         \node (5) at (1.5,-2) [nodeSingleton,label=below:\footnotesize{ $\{5\}$}] {};

         \draw[path,dashed,ultra thick] (345) -- (4);
         \draw[path,dashed,ultra thick] (345) -- (5);
         \draw[path,dashed,ultra thick] (456) -- (4);
         \draw[path,dashed,ultra thick] (456) -- (5);
      \end{tikzpicture}
      \caption{%
         Monomial $\setdef{4,5}$ is not present.
      }
      \label{fig:tripleMonomialMissing}
   \end{subfigure}
   \begin{subfigure}[b]{0.55\textwidth}
      \centering
      \begin{tikzpicture}
         \node (m1) at (5.5,0) [nodeMonomial,label=above:\footnotesize{ $m_1 = \{1,2,3,4\}$}] {};
         \node (m3) at (8,0) [nodeMonomial,label=above:\footnotesize{ $m_3 = \{3,4,5\}$}] {};
         \node (m2) at (10.5,0) [nodeMonomial,label=above:\footnotesize{ $m_2 = \{4,5,6\}$}] {};

         \node (m13) at (7,-1) [nodeMonomial,label=right:\footnotesize{ $\{3,4\}$}] {};
         \node (m23) at (9,-1) [nodeMonomial,label=right:\footnotesize{ $\{4,5\}$}] {};
         \node (1) at (5.5,-2) [nodeSingleton,label=below:\footnotesize{ $\{1\}$}] {};
         \node (2) at (6.5,-2) [nodeSingleton,label=below:\footnotesize{ $\{2\}$}] {};
         \node (3) at (7.5,-2) [nodeSingleton,label=below:\footnotesize{ $\{3\}$}] {};
         \node (4) at (8.5,-2) [nodeSingleton,label=below:\footnotesize{ $\{4\}$}] {};
         \node (5) at (9.5,-2) [nodeSingleton,label=below:\footnotesize{ $\{5\}$}] {};
         \node (6) at (10.5,-2) [nodeSingleton,label=below:\footnotesize{ $\{6\}$}] {};

         \draw[arc] (m1) -- (1);
         \draw[arc] (m1) -- (2);
         \draw[arc] (m1) -- (m13);
         \draw[arc] (m2) -- (m23);
         \draw[arc] (m2) -- (6);
         \draw[arc,dashed,ultra thick] (m3) -- (m13);
         \draw[arc,dashed,ultra thick] (m3) -- (m23);
         \draw[arc] (m13) -- (3);
         \draw[arc,dashed,ultra thick] (m13) -- (4);
         \draw[arc,dashed,ultra thick] (m23) -- (4);
         \draw[arc] (m23) -- (5);
      \end{tikzpicture} 
      \caption{%
         Monomials $\setdef{3,4}$ and $\setdef{4,5}$ are present.
      }
      \label{fig:tripleMonomialsPresent}
   \end{subfigure}
   \caption{%
      Illustration of different situations in \cref{ex:tripleMonomial}.
      Induced cycles of $G(D(\linearization))$ are shown dashed.
   }
   \label{fig:tripleMonomial}
\end{figure}

The outline of the proof of \cref{thm:RIP} is as follows.
First, we show for an integeral simple linearization and distinct monomials~$m_i$ and~$m_j$ that~$m_i \cap m_j$ is a common successor of~$m_i$ and~$m_j$ if~$m_i \cap m_j \neq \varnothing$ as illustrated in \cref{ex:tripleMonomial}.
This allows to prove that every linearization graph contains a cycle if~$\targetMonomials$ does not have the monomial intersection property.
For \ref{MIP1} this proof was already sketched in \cref{ex:tripleMonomial};
for \ref{MIP2} we show that there exists a cycle by concatenating paths connecting~$m_i$ and~$m_{i+1}$
with their intersection~$m_{i} \cap m_{i+1}$.
Afterwards, we construct a particular linearization that is acyclic if~$\targetMonomials$ has the monomial intersection property.

\subsection{Properties of \texorpdfstring{$G(D(\linearization))$}{G(D(L))}}

Throughout this section, we assume that~$\targetMonomials$ is a set of monomials with singletons $\singletonMonomials$ and that~$\linearization = (n, \allMonomials, \constraints)$ is a simple linearization with $\singletonMonomials \cup \targetMonomials \subseteq \allMonomials$.
Moreover, we assume that every node of $D(\linearization)$ with in-degree~$0$ belongs to $\targetMonomials$.
This assumption is without loss of generality since \cref{thm:redundantMonomials} guarantees that the removal of nodes $m \in \properMonomials \setminus \targetMonomials$ with in-degree~$0$ does not affect integrality.
Hence, only Property~\ref{PropertyCycle} of \cref{thm:simpleLinearizationProjectionIntegral} is relevant for integrality, since Property~\ref{PropertyPathAbove} is satisfied automatically.

All successor and predecessor relations mentioned in the following results are w.r.t.\ the linearization graph~$D(\linearization)$.

We now prove a key lemma for understanding the construction of linearization whose digraph is acyclic. It shows that the intersection of two target monomials $m_1$ and $m_2$ has to be part of an acyclic linearization.

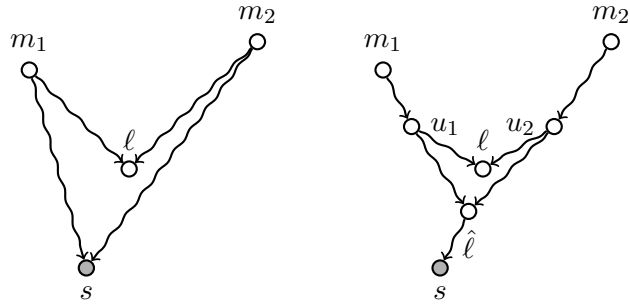
\begin{figure}[tb]
  \begin{tikzpicture}[scale=0.75]
    \node[nodeMonomial, label={above:$m_1$}] at (1,4) (m1) {};
    \node[nodeMonomial, label={above:$m_2$}] at (5,4.5) (m2) {};
    \node[nodeMonomial, label={above:$\ell$}] at (2.75,2.25) (l) {};
    \node[nodeSingleton, label={below:$s$}] at (2,0.5) (s) {};
    
    \draw[path] (m1)--(l);
    \draw[path] (m2)--(l);
    \draw[path] (m1)--(s);
    \draw[path] (m2)--(s);
  \end{tikzpicture}
  \qquad
  \begin{tikzpicture}[scale=0.75]
    \node[nodeMonomial, label={above:$m_1$}] at (1,4) (m1) {};
    \node[nodeMonomial, label={above:$m_2$}] at (5,4.5) (m2) {};
    \node[nodeMonomial, label={right:$u_1$}] at (1.5,3) (u1) {};
    \node[nodeMonomial, label={left:$u_2$}] at (4,3) (u2) {};
    \node[nodeMonomial, label={above:$\ell$}] at (2.75,2.25) (l) {};
    \node[nodeMonomial, label={below:$\hat{\ell}$}] at (2.5,1.5) (lh) {};
    \node[nodeSingleton, label={below:$s$}] at (2,0.5) (s) {};

    \draw[path] (m1)--(u1);
    \draw[path] (m2)--(u2);
    \draw[path] (u1)--(l);
    \draw[path] (u2)--(l);
    \draw[path] (u1)--(lh);
    \draw[path] (u2)--(lh);
    \draw[path] (lh)--(s);
  \end{tikzpicture}
  \caption{Construction from the proof of \cref{thm:intersectionLinearization}.}
  \label{fig:subsetLinearization}
\end{figure}

\begin{lemma}\label{thm:intersectionLinearization}
  Let~$m_1$, $m_2 \in \allMonomials$ with~$m_1 \cap m_2 \neq \varnothing$ and $m_1 \cap m_2 \notin \succ(m_1) \cap \succ(m_2)$.
  Then~$G(D(\linearization))$ contains a cycle.
\end{lemma}

\begin{proof}
  Let~$m_1$, $m_2 \in \targetMonomials$ satisfy the requirements of the lemma.
  We first observe that if~$\card{m_1 \cap m_2} = 1$, \cref{thm:consistencyPaths} shows that there exists a path from $m_1$ and $m_2$ to the singleton in $m_1 \cap m_2$,
  contradicting the assumption that $m_1 \cap m_2 \notin \succ(m_1) \cap \succ(m_2)$.
  Thus, $\card{m_1 \cap m_2} \geq 2$.

  Consider $\ell \in \allMonomials$ of largest cardinality with $\ell \in \succ(m_1) \cap \succ(m_2)$.
  Then $\ell \subset m_1 \cap m_2$, because $\ell$ cannot contain singletons that are neither contained in $m_1$ nor in $m_2$; moreover, $\ell \neq m_1 \cap m_2$ by assumption. 
  Thus, there exists a singleton $s \in (m_1 \cap m_2) \setminus \ell$ (using $\card{m_1 \cap m_2} \geq 2$).
  Then there exist distinct paths $P(m_1,\ell)$, $P(m_2,\ell)$, $P(m_1,s)$, $P(m_2,s)$; see \cref{fig:subsetLinearization} for an illustration.

  Let $u_1$ and $u_2$ be the last common node of $P(m_1,\ell)$ and $P(m_1,s)$ as well as $P(m_2,\ell)$ and $P(m_2,s)$.
  Since $s \in u_1$, $s \in u_2$, neither $u_1$ nor $u_2$ is equal to $\ell$.
  Moreover, $\ell \subset u_1$, $\ell \subset u_2$ and thus $u_1$ and $u_2$ are distinct---otherwise $\ell$ would not be maximal.
  Furthermore, let $\hat{\ell}$ be the first common node of the paths $P(u_1,s)$ and $P(u_2,s)$.
  Then $\hat{\ell}$ is not equal to the other three nodes.
  In total there is a cycle with upper nodes~$u_1$ and~$u_2$ as well as lower nodes~$\ell_1$ and~$\ell_2$.
\end{proof}

Note that $m_1 \subset m_2$ is allowed in \cref{thm:intersectionLinearization} ($m_1 \notin \succ(m_2)$ in this case).

In summary, this shows that every integral linearization of~$\targetMonomials$ uses the monomials~$m_1 \cap m_2$ if $m_1$, $m_2 \in \allMonomials$ and~$m_1 \cap m_2 \neq \varnothing$ hold due to \cref{thm:intersectionLinearization} and \cref{thm:simpleLinearizationProjectionIntegral}.
This allows us to prove one direction of \cref{thm:RIP}:
the next lemma shows that every integral linearization cannot contain monomials that satisfy~\ref{MIP1}.
Then \cref{thm:nonlaminarFractional2} will show that they cannot satisfy~\ref{MIP2}.

\begin{lemma}
   \label{thm:nonlaminarFractional}
   Let~$m_1$, $m_2$, $m_3 \in \targetMonomials$ be pairwise different monomials satisfying Property~\ref{MIP1}, i.e., their intersection is nonempty and $m_{3} \cap (m_{1} \cup m_{2})$ is a proper superset of~$m_3 \cap m_{1}$ and~$m_3 \cap m_{2}$.
   Then $G(D(\linearization))$ contains a cycle.
\end{lemma}

\begin{proof}
   Since~$m_3 \cap (m_{1} \cup m_{2})$ is a proper superset of~$m_3 \cap m_{1}$ and~$m_3 \cap m_{2}$, neither~$m_{1}$ is a subset of~$m_{2}$ and~$m_3$ nor~$m_{2}$ is a subset of~$m_{1}$ and~$m_3$.
   Consequently,
   \[
      m_3
      \supset
      m_3 \cap m_{1}
      \supset m_{3} \cap m_{1} \cap m_{2}
      \qquad\text{and}\qquad
      m_3
      \supset
      m_3 \cap m_{2}
      \supset
      m_{3} \cap m_{1} \cap m_{2}.
   \]
   If a linearization~$\linearization$ does not make use of one of the intersections~$m_{3} \cap m_{1}$, $m_{3} \cap m_{2}$, or $m_3 \cap m_{1} \cap m_{2}$, \cref{thm:intersectionLinearization} shows that~$G(D(\linearization))$ contains a cycle.
   For this reason, assume all these intersections appear as monomials
   in~$\linearization$.
   For~$i \in [2]$, let~$P_i$ be a path connecting~$m_3$ and~$m_1 \cap m_2
   \cap m_3$ that uses~$m_3 \cap m_i$ as intermediate node.
   Then~$m_3 \cap m_1$ is not a successor of~$m_3 \cap m_2$, because otherwise, $m_3 \cap (m_1 \cup m_2) = m_3 \cap m_2$, contradicting the assumption.
   Analogously we find that~$m_3 \cap m_2$ is not a successor of~$m_3 \cap m_1$.
   Thus, $P_1$ and~$P_2$ are different.
   Using the same arguments as in the proof of
   \cref{thm:intersectionLinearization}, we find that~$P_1 \cup P_2$ contains a
   cycle whose upper node is the first common predecessor of~$m_3 \cap
   m_{1}$ and~$m_{3} \cap m_{2}$ and whose lower node is the first common
   successor of these two intersections.
\end{proof}

\begin{figure}[tb]
  \begin{tikzpicture}
    \node[nodeMonomial, label={above:$m_i$}] at (1,3) (mi) {};
    \node[nodeMonomial, label={left:$u_i$}] at (1,2) (ui) {};
    \node[nodeMonomial, label={above:$m_{i+1}$}] at (3,3) (mip) {};
    \node[nodeMonomial, label={right:$u_{i+1}$}] at (3,2) (uip) {};
    \node[nodeMonomial, label={below:$\ell_{i-1}$}] at (0,1) (lim) {};
    \node[nodeMonomial, label={below:$\ell_i$}] at (2,1) (li) {};
    \node[nodeMonomial, label={below:$\ell_{i+1}$}] at (4,1) (lip) {};

    \draw[path] (-1,2)--(lim);
    \draw[path] (mi)--(ui);
    \draw[path] (ui)--(lim);
    \draw[path] (ui)--(li);
    \draw[path] (mip)--(uip);
    \draw[path] (uip)--(li);
    \draw[path] (uip)--(lip);
    \draw[path] (5,2)--(lip);
  \end{tikzpicture}
  \caption{Construction from the proof of \cref{thm:nonlaminarFractional2}.}
  \label{fig:nonlaminarFractional2}
\end{figure}
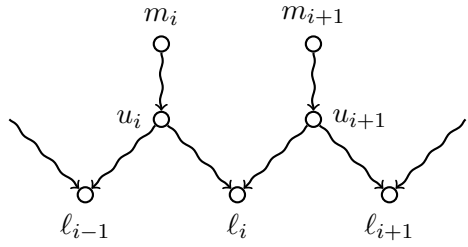

\begin{lemma}
   \label{thm:nonlaminarFractional2}
   Let~$m_1, \dots, m_k \in \targetMonomials$ with $k \geq 3$ be pairwise different monomials satisfying Property~\ref{MIP2}, i.e., for all~$i$, $j \in [k]$, we have~$m_i \cap m_j \neq \varnothing$ if and only if~$i$ and~$j$ differ by at most~$1$ modulo~$k$.
   Then $G(D(\linearization))$ contains a cycle.
\end{lemma}

\begin{proof}
  Let~$m$, $m' \in \targetMonomials$.
  If~$m \cap m' \neq \varnothing$, by \cref{thm:intersectionLinearization}, we can assume that~$D(\linearization)$ contains the node~$m \cap m'$, which is a successor of~$m$ and~$m'$.
  Therefore, for every~$i \in [k]$, $\ell_i \define m_i \cap m_{i+1}$ is a node of~$D(\linearization)$ and it is a successor of~$m_i$ and~$m_{i+1}$, where $m_{k+1} \define m_1$.
  (Note that $\ell_i$ might coincide with either $m_i$ or $m_{i+1}$ if~$k = 3$.)
  Thus, for every~$i \in [k]$, there exist directed paths~$P(m_i,\ell_{i})$ and~$P(m_{i+1}, \ell_{i})$ in~$D(\linearization)$, see \cref{fig:nonlaminarFractional2}.

  For $i \in [k]$, let~$u_i$ be the last common node of~$P(m_i,\ell_{i})$ and~$P(m_i,\ell_{i-1})$, where we set $\ell_0 \define m_k \cap m_1$.
  Consider the paths $P(u_i,\ell_{i})$ and~$P(u_{i+1}, \ell_{i})$.
  These paths are arc-disjoint, because~$\ell_{i}$ contains all elements that appear in both~$m_i$ and~$m_{i+1}$, and thus, all elements that appear in~$u_i$ and~$u_{i+1}$.
  The union of the paths~$P(u_i, \ell_{i})$, $P(u_i, \ell_{i-1})$ for every~$i \in [k]$ yields a cycle in~$G(D(\linearization))$.
\end{proof}

Combining the previous results, establishes the necessity of the
requirements in \cref{thm:RIP} for the existence of integral linearizations.

\begin{proof}[Necessity proof for \cref{thm:RIP}.]
   Let $\targetMonomials$ be a set of monomials with singletons $\singletonMonomials$ such that at least one of the properties~\ref{MIP1} and~\ref{MIP2} is satisfied.
   Assume, for the sake of contradiction that there exists a simple linearization $\linearization = (n, \allMonomials, \constraints)$ with $\targetMonomials \subseteq \properMonomials$ such that $\proj[\singletonMonomials \cup \targetMonomials][P(\linearization)]$ is integral.
   By \cref{thm:nonlaminarFractional} or \cref{thm:nonlaminarFractional2}, $G(D(\linearization))$ contains a subgraph $Z$ of $D(\linearization)$ such that $G(Z)$ is a cycle.
   By \cref{thm:redundantMonomials}, we can assume that each monomial $m \in \properMonomials$ with in-degree~$0$ belongs to $\targetMonomials$ since we could otherwise construct a linearization without this monomial retaining the same properties.
   Hence, $Z$ satisfies Property~\ref{PropertyPathAbove} of \cref{thm:simpleLinearizationProjectionIntegral}, contradicting the integrality assumption.
\end{proof}

\subsection{Constructing Acyclic Linearizations}
\label{sec:constructingAcyclicLinearizations}

For the sufficiency proof for \cref{thm:RIP} we explicitly construct a simple linearization that is acyclic, provided~$\targetMonomials$ has the monomial intersection property.
For a given family of target monomials~$\targetMonomials$ and singletons~$\singletonMonomials$, define the sets
\[
   \nonemptyIntersection(\targetMonomials)
   \define
   \Big\{
      I \subseteq \targetMonomials
      \suchthat
      \bigcap_{m \in I} m \neq \varnothing
   \Big\}
   \quad
   \text{and}
   \quad
   \allMonomials' = \allMonomials'(\targetMonomials)
   \define
   \Big\{
      \bigcap_{m \in I} m
      \suchthat
      I \in \nonemptyIntersection(\targetMonomials)
   \Big\}.
\]
The set~$\nonemptyIntersection(\targetMonomials)$ contains all subsets of monomials whose intersection is non-empty, whereas~$\allMonomials'$ contains the corresponding non-empty intersections.
Notice that $\allMonomials'$ may contain singletons.

\begin{remark}
  Note that the set $\nonemptyIntersection(\targetMonomials)$ forms an
  independence system (or simplicial complex). It is sometimes called the
  \emph{nerve} of the family $\targetMonomials$, see, e.g.,
  Bj\"orner~\cite{Bjo81}.
  The set $\allMonomials'$ ordered by inclusion forms a partially ordered
  set (poset), which is closed under intersection. In fact, if we add an
  artificial maximal element~$\hat{1}$ and a minimal element $\varnothing$,
  then it is a lattice. A relevant operation in this context is the
  \emph{closure} of a set of variables (singletons)
  $s \subseteq \singletonMonomials$:
  $\cl(s) \define \bigcap \{m \in \targetMonomials \suchthat s \subseteq m\}$
  (if no $m$ with $s \subseteq m$ exists, we return $\hat{1}$); see Caspard
  and Monjardet~\cite{CasM03} for an overview of such structures. This can
  be exploited algorithmically to construct $\allMonomials'$, see
  \cite{KaibelPfetsch2002}.
\end{remark}

Using the above sets, we define $\allMonomials^\star \define \allMonomials' \cup
\singletonMonomials$ and the digraph~$D^{\star} = (\allMonomials^\star, A^{\star})$, where for~$m_1$, $m_2 \in \allMonomials^\star$, the set~$A^\star$ contains an arc~$(m_1, m_2)$ if and only if~$m_2 \subset m_1$ and there does not exist~$m_3 \in \allMonomials^\star$ satisfying~$m_2 \subset m_3 \subset m_1$.
That is, the graph~$D^\star = (\allMonomials^\star, A^\star)$ represents the Hasse diagram of~$\allMonomials^\star$, endowed with the subset relation as partial order.
We denote the undirected version of~$D^\star$ by~$G^\star$ and define, for
every node~$m \in \allMonomials^\star$, $\delta^+(m) \define \{m' \in \allMonomials^\star \suchthat (m,m') \in A^\star\}$.

Since~$\singletonMonomials \cup \targetMonomials \subseteq \allMonomials^\star$ and~$m = \bigcup_{m' \in \delta^+(m)} m'$ for every~$m \in \allMonomials^\star \setminus \singletonMonomials$, there exists a linearization~$\poset = (n, \allMonomials^\star, \constraints^\star)$ of~$\targetMonomials$ such that~$D^\star$ is the digraph associated with~$\poset$.
\Cref{fig:runningExamplePoset} shows the linearization graph of~$\poset$ for the running example.

\begin{figure}[tb]
   \centering
   \begin{tikzpicture}
      \node[nodeSingleton,label=below:\footnotesize{$\{1\}$}] (1) at (0,0) {};
      \node[nodeSingleton,label=below:\footnotesize{$\{2\}$}] (2) at (2,0) {};
      \node[nodeSingleton,label=below:\footnotesize{$\{3\}$}] (3) at (4,0) {};
      \node[nodeSingleton,label=below:\footnotesize{$\{4\}$}] (4) at (6,0) {};
      \node[nodeSingleton,label=below:\footnotesize{$\{5\}$}] (5) at (8,0) {};
      \node[nodeSingleton,label=below:\footnotesize{$\{6\}$}] (6) at (10,0) {};

      \node[nodeMonomial,label=right:\footnotesize{$\{3,4\}$}] (34) at (5,1) {};
      \node[nodeMonomial,label=right:\footnotesize{$\{4,5\}$}] (45) at (7,1) {};

      \node[nodeTarget,label=left:\footnotesize{$\{1,2,3,4\}$}] (1234) at (2,2) {};
      \node[nodeTarget,label=left:\footnotesize{$\{3,4,5\}$}] (345) at (6,2) {};
      \node[nodeTarget,label=left:\footnotesize{$\{4,5,6\}$}] (456) at (8,2) {};

      \foreach \a/\b in {%
         1234/1, 1234/2, 1234/34,
         345/34, 345/45,
         456/45, 456/6,
         34/3, 34/4,
         45/4, 45/5%
      }{
         \draw[arc] (\a) -- (\b);
      }
  \end{tikzpicture}
  \caption{The linearization graph~$D^{\star}$ for \cref{ex:runningExample}.
  \label{fig:runningExamplePoset}}
\end{figure}

\begin{lemma}
  \label{lem:posetCycle1}
  If $D^\star$ contains a subgraph $Z$ for which $G(Z)$ is a cycle and $\card{\upperNodes(Z)} = 1$, then~\ref{MIP1} holds.
\end{lemma}

\begin{proof}
   Suppose there exists a subgraph~$Z$ of~$D^{\star}$ such that~$G(Z)$ is a cycle and~$\card{\upperNodes(Z)} = 1$.
   We can assume that among all such subgraphs the unique upper node of~$Z$ has maximum cardinality.
   Then there exist monomials~$u$, $\ell \in \allMonomials^\star$ such that~$\upperNodes(Z) = \setdef{u}$ and~$\lowerNodes(Z) = \setdef{\ell}$.
   Moreover, there are two distinct arc-disjoint paths~$P_1$ and~$P_2$ in~$D^\star$ connecting~$u$ and~$\ell$.
   Note that both paths consist of at least two arcs: On the one hand, we cannot have $\card{P_1} = \card{P_2} = 1$, since this would imply~$P_1 = P_2$.
   On the other hand, if one path, say~$P_1$, consists of a single arc, no~$m \in \allMonomials^\star$ exists with~$\ell \subset m \subset u$, contradicting the existence of~$P_2$ by definition of the arc set $A^\star$.

   Let~$v_1$ and~$v_2$ be the first successor of~$u$ in~$P_1$ and~$P_2$, respectively.
   Choose a target monomial $t_3 \in \pred(u) \cap \targetMonomials$, which exists by definition of $D^{\star}$.
   Moreover, let $s_1 \in v_1 \setminus v_2$ and~$s_2 \in v_2 \setminus v_1$, noting that both sets are non-empty since $v_1$ and~$v_2$ are different children of~$u$ and neither is a successor of the other.
   Finally, choose $t_1 \in \pred(v_1) \cap \targetMonomials$ such that $s_2 \notin t_1$.
   This is possible, since if each monomial in $\pred(v_1) \cap \targetMonomials$ contained $s_2$, this would have implied $s_2 \in v_1$ as well, contradicting the choice of $s_2$.
   Similarly, choose $t_2 \in \pred(v_2) \cap \targetMonomials$ such that $s_1 \notin t_2$.

   This implies that $\ell \in \succ(t_1) \cap \succ(t_2) \cap \succ(t_3)$, from which we obtain $t_1 \cap t_2 \cap t_3 \neq \varnothing$.
   Furthermore, $s_1 \in t_3 \cap (t_1 \setminus t_2)$ and $s_2 \in t_3 \cap (t_2 \setminus t_1)$ establish \ref{MIP1} for $t_1$, $t_2$ and $t_3$, which concludes the proof.
\end{proof}

\begin{lemma}
  \label{lem:posetCycle2}
  $D^\star$ does not contain a subgraph $Z$ for which $G(Z)$ is a cycle and $\card{\upperNodes(Z)} = 2$.
\end{lemma}

\begin{proof}
  Suppose there exists a subgraph~$Z$ of~$D^{\star}$ such that~$G(Z)$ is a cycle and~$\card{\upperNodes(Z)} = 2$.
  Let~$u_1$ and~$u_2$ be the upper nodes and~$\ell_1$ and~$\ell_2$ be the lower nodes of~$Z$.
  Then, by the definition of $\upperNodes(Z)$ and $\lowerNodes(Z)$, there exist pairwise internally node disjoint paths~$P^i_j$ connecting~$u_i$ and~$\ell_j$ for~$i$, $j \in [2]$ whose union yields~$Z$.
  Since $u_1 \cap u_2 \supseteq \ell_1 \neq \varnothing$, $u_1 \cap u_2$ is a node of~$D^\star$.

  Observe that both~$P^1_1$ and~$P^1_2$ have to traverse~$u_1 \cap u_2$:
  Assume that this is not the case. Then, for~$i \in [2]$, let~$\bar{u}_i$ be the smallest superset of~$u_1 \cap u_2$ in~$P^1_i$ and~$\bar{\ell}_i$ be the largest subset of~$u_1 \cap u_2$ in~$P^1_i$;
  these nodes exist because $u_i$ or $\ell_i$ are candidates, respectively.
  This implies~$\bar{\ell}_i \subset u_1 \cap u_2 \subset \bar{u}_i$. Thus, $G^{\star}$ does not contain the edge~$\{\bar{u}_i, \bar{\ell}_i\}$.
  Hence, both~$P^1_1$ and~$P^1_2$ have to use~$u_1 \cap u_2$ as an intermediate node.
  This, however, is a contradiction to the internal node disjointness of~$P^1_1$ and~$P^1_2$,
  and the result follows.
\end{proof}

\begin{lemma}
  \label{lem:posetCycle3}
  Let $k$ be the smallest integer such that $D^\star$ contains a subgraph $Z$ for which $G(Z)$ is a cycle and $\card{\upperNodes(Z)} = k$.
  If $k \geq 3$, then~\ref{MIP2} holds.
\end{lemma}

\begin{proof}
  Suppose that~$Z$ is chosen such that $G(Z)$ is a cycle and $k \define \card{\upperNodes(Z)} \geq 3$ is minimal.
  Let~$\upperNodes(Z) = \{u_1, \dots, u_k\}$
  and~$\lowerNodes(Z) = \{\ell_1, \dots, \ell_k\}$ be such that there
  exist pairwise internally node disjoint paths in~$D^\star$ connecting~$u_i$
  and~$\ell_i$ as well as~$u_i$ and~$\ell_{i+1}$, where we set $\ell_{k+1} = \ell_1$.

  For every~$i \in [k]$, let~$t_i \in \targetMonomials$ be such
  that~$u_i \in \succ(t_i)$. Oberserve that~$t_i \notin \succ(t_j)$ for all
  distinct~$i$, $j \in [k]$: Otherwise, one of the following two cases
  occurs.
  \begin{enumerate*}[label=(\roman{*})]
  \item There exist two distinct paths from~$t_j$ to~$u_i \cap u_j$,
    one using~$u_i$ and the other~$u_j$ as intermediate node; note that $u_i \cap u_j$ is a node by definition of~$D^\star$. The
    union of these paths contains a cycle with a single upper node,
    contradicting the minimality assumption on~$Z$.
  \item Every path in~$D^\star$
    from~$t_j$ to~$u_i \cap u_j$ that traverses~$u_i$ also traverses~$u_j$;
    w.l.o.g.\ $u_i$ is a proper successor of $u_j$.
    Thus, there exists a path $P(u_j,u_i)$ not completely contained in~$Z$, since otherwise~$u_i$ cannot be an upper node.
    This shows that there exists a ``shortcut'', which produces a cycle with less upper nodes, a contradiction to the minimality of~$Z$.
  \end{enumerate*}

  Next, observe that~$t_i \cap t_j \neq \varnothing$ if~$i$ and~$j$ differ
  by at most~1 (modulo~$k$), because~$u_i \subseteq t_i$
  and~$u_j \subseteq t_j$ and~$u_i \cap u_j \neq \varnothing$ in this
  case.

  Thus, it remains to show that~$t_i \cap t_j = \varnothing$ if~$i$
  and~$j$ differ by at least~$2$ (modulo~$k$).
  If this was false, there would exist~$i$, $j \in [k]$ that differ by at
  least~$2$ (modulo~$k$) such that~$t_i \cap t_j \neq \varnothing$.
  Then, there exist two internally node disjoint paths~$P_i$ and~$P_j$
  connecting~$t_i$ and~$t_i \cap t_j$ as well as~$t_j$ and~$t_i \cap t_j$,
  respectively. Consider the subgraph~$G'$ of~$G^\star$ that is obtained by
  \begin{itemize}
  \item starting a walk in~$t_i$,
  \item following~$P_i$ to reach~$t_i \cap t_j$,
  \item going along~$P_j$ in reverse order to approach~$t_j$,
  \item following a path to reach~$u_j$ (which exists since~$u_j \in
    \succ(t_j)$),
  \item continuing along a path contained in~$Z$ to reach~$u_i$, and
  \item going back to~$t_i$ via a (reversed) path
    between~$t_i$ and~$u_i$ (which exists since~$u_i \in
    \succ(t_i)$).
  \end{itemize}
  If~$G'$ is a cycle, then~$\card{\upperNodes(G')} < k$, because the
  upper nodes~$u_i$ and~$u_j$ in~$Z$ are replaced by the upper
  nodes~$t_i$ and~$t_j$ and the remaining nodes
  in~$\upperNodes(G')$ form a proper subset of~$\upperNodes(Z)$, since we
  traverse only one path connecting~$u_i$ and~$u_j$ in~$Z$ passing through at least one other node in $\upperNodes(Z)$.
  If~$G'$ is not a cycle, there exists a short cut~$Z'$ in~$G'$ being a cycle and
  fulfilling~$\card{\upperNodes(Z')} < k$ by the same arguments. Thus, we
  obtain a contradiction to the minimality assumption on~$Z$. As a
  consequence, $t_i \cap t_j \neq \varnothing$ if and only
  if~$i$ and~$j$ differ by at most~$1$, which finally shows the assertion.
\end{proof}

Combining \cref{lem:posetCycle1,lem:posetCycle2,lem:posetCycle3}, we have proved:

\begin{corollary}
  \label{thm:posetCycle}
  If~$G^\star$ contains a cycle, then~\ref{MIP1} or~\ref{MIP2} hold.
\end{corollary}

We are now ready to provide the missing part of the proof of \cref{thm:RIP}.

\begin{proof}[Sufficiency proof for \cref{thm:RIP}]
   Assume that~$\targetMonomials$ has the monomial intersection property.
   Let $\poset$ be the linearization described above.
   Then \cref{thm:posetCycle} shows that~$G^{\star}$ does not contain a cycle.
   Hence, the linearization graph $G(D(\poset))$ is acylic, showing that $\proj[\singletonMonomials \cup \targetMonomials](P(\poset))$ is integral by \cref{thm:simpleLinearizationProjectionIntegral}.
\end{proof}

\begin{remark}\label{rem:PropertyIntersectionPoset}
  A frequently used property in polynomial optimization is the running
  intersection property, see, e.g., Lasserre~\cite{Lasserre2006} and Kojima
  and Muramatsu~\cite{KojimaMuramatsu2009}.
  A set of monomials~$\targetMonomials$ has the
  running intersection property if there exists an
  ordering~$t_1, \dots, t_k$ of the sets in~$\targetMonomials$ such that for
  each~$j \in \{2, \dots, k\}$ there exists~$i \in [j-1]$ such that
  $t_j \cap \bigcup_{r = 1}^{j-1} t_r = t_j \cap t_i$.  While the running
  intersection property excludes~\ref{MIP2}, cf.\ Beeri et
  al.~\cite[Theorem~3.4(1)]{BeeriEtAl1983}, it is still possible
  that~\ref{MIP1} occurs, see \cref{ex:runningExample} with~$t_1 =
  \{1,2,3,4\}$, $t_2 = \{3,4,5\}$, and~$t_3 = \{4,5,6\}$.
\end{remark}

\section{Algorithmic Consequences}
\label{SectionAlgorithmicConsequences}

Consider a set $\targetMonomials$ that satisfies the monomial intersection property, i.e., neither \ref{MIP1} nor \ref{MIP2} hold.
In this section we discuss three approaches to solve an instance of~\eqref{eq:polyMin} in polynomial time.
While the first approach is due to Del Pia und Khajavirad~\cite{DelPiaKhajavirad2018}, the two further approaches follow from \cref{sec:integrality} and \cref{sec:existence}.

\subsection{Relation to Flower Inequalities}

Del Pia and Khajavirad~\cite{DelPiaKhajavirad2018} derived a complete
linear description of the multilinear polytope provided the intersections
of target monomials have a certain structure.
Using our notation, their result reads as follows.

\begin{theorem}[Del Pia and Khajavirad~\cite{DelPiaKhajavirad2018}]
  Let~$\targetMonomials$ be a set of monomials in variables~$x_1, \dots,
  x_n$.
  Then, $\BMP_n(\targetMonomials)$ is completely described by box constraints, the
  inequalities of the standard linearization, as well as a (potentially)
  exponentially large class of so-called flower inequalities, which can be
  separated in time~$\orderO(n^2 \card{\targetMonomials}^2 + n
  \card{\targetMonomials}^3)$, if and only if neither of the following two
  conditions hold:
  \begin{enumerate}[label=(\Alph*$'$)]
  \item
    \label{DelPia1}
    there exist pairwise distinct~$i_1, i_2, i_3 \in [n]$ such
    that
    \[
    \big\{ \{i_1, i_3\}, \{i_2, i_3\}, \{i_1, i_2, i_3\} \big\}
    \subseteq
    \{m \cap \{i_1, i_2, i_3\} \suchthat m \in \targetMonomials\},
    \]

  \item
    \label{DelPia2}
    there exist pairwise different~$i_1, \dots, i_k \in [n]$ and
    pairwise different~$m_1, \dots, m_k \in \targetMonomials$ with~$k \geq
    3$ such that
    \begin{itemize}
    \item $i_j \in m_j, m_{j+1}$ for every~$j \in [k-1]$,

    \item $i_k \in m_1 \cap m_k$, and

    \item $i_j$, $j \in [k]$, is not contained in any further target monomial.
    \end{itemize}
  \end{enumerate}
\end{theorem}

In the following, we show that~$\targetMonomials$ satisfies
Conditions~\ref{DelPia1} and~\ref{DelPia2} if and only if it satisfies
Properties~\ref{MIP1} and~\ref{MIP2}.
Thus, the results presented in this article complement the results by Del
Pia and Khajavirad~\cite{DelPiaKhajavirad2018}, because its now possible to
avoid the exponentially many inequalities in the description of~$\BMP_n(\targetMonomials)$ by
using a polynomial size extended formulation that is given by~$P(\poset)$.
In fact, we will see in the next section that the extended formulation has
linear size provided Properties~\ref{MIP1} and~\ref{MIP2}
(resp.\ref{DelPia1} and~\ref{DelPia2}) hold.

\begin{lemma}
  Let~$\targetMonomials$ be a set of monomials in variables~$x_1, \dots,
  x_n$ that satisfies Property~\ref{MIP1} or~\ref{MIP2}.
  Then, Condition~\ref{DelPia1} or~\ref{DelPia2} holds as well.
\end{lemma}

\begin{proof}
  Suppose~$\targetMonomials$ satisfies Property~\ref{MIP1}, where~$m_1, m_2,
  m_3 \in \targetMonomials$ are the corresponding monomials.
  Let~$i_3 \in m_1 \cap m_2 \cap m_3$ and~$i_1 \in m_1 \setminus m_2$ as
  well as~$i_2 \in m_2 \setminus m_1$.
  The elements~$i_1$, $i_2$, $i_3$ are distinct and exist, because~$m_1
  \cap m_2 \cap m_3$ is non-empty and $m_3 \cap (m_1 \cup m_2)$ is a
  proper superset of both~$m_3 \cap m_1$ and~$m_3 \cap m_2$.
  Consequently, $\big\{ \{i_1, i_3\}, \{i_2, i_3\}, \{i_1, i_2, i_3\}
  \big\} \subseteq \big\{m \cap \{i_1, i_2, i_3\} \suchthat m \in \{m_1,
  m_2, m_3\}\big\}$, showing Condition~\ref{DelPia1}.

  If~$\targetMonomials$ satisfies Property~\ref{MIP2} for monomials~$m_1,
  \dots m_k$, $k \geq 3$, then~$m_i \cap m_j \neq \varnothing$ if and only
  if~$i$ and~$j$ differ by at most~$1$ modulo~$k$.
  Hence, no element in~$m_i \cap m_j$ can be contained in a further target
  monomial, which implies Condition~\ref{DelPia2}.
\end{proof}

\begin{lemma}
  Let~$\targetMonomials$ be a set of monomials in variables~$x_1, \dots,
  x_n$ that satisfies Condition~\ref{DelPia1} or~\ref{DelPia2}.
  Then, Property~\ref{MIP1} or~\ref{MIP2} hold as well.
\end{lemma}

\begin{proof}
  Suppose~$\targetMonomials$ satisfies Condition~\ref{DelPia1} for~$i_1$,
  $i_2$, $i_3 \in [n]$.
  Then, there exist~$m_1$, $m_2$, $m_3 \in \targetMonomials$ such that
  their restrictions to~$\{m_1, m_2, m_3\}$ are~$\{i_1, i_3\}$, $\{i_2,
  i_3\}$, and~$\{i_1, i_2, i_3\}$.
  Since~$i_3$ is contained in all these monomials, $m_1 \cap m_2 \cap m_3
  \neq \varnothing$.
  Moreover, $m_3 \cap (m_1 \cup m_2)$ contains~$i_2$, which is not
  contained in~$m_1$.
  Hence, $m_3 \cap (m_1 \cup m_2)$ is a proper superset of~$m_3 \cap m_1$.
  Analogously one shows~$m_3 \cap (m_1 \cup m_2) \supset m_3 \cap m_2$,
  which implies Property~\ref{MIP1}.

  If~$\targetMonomials$ satisfies Condition~\ref{DelPia2}, choose the
  corresponding elements~$i_1, \dots, i_k$ and monomials~$m_1, \dots, m_k$
  such that~$k$ is minimal.
  If~$m_i \cap m_j \neq \varnothing$ if and only if~$\card{i - j} \leq 1$
  for every~$i$, $j \in [k]$, Property~\ref{MIP2} follows immediately.
  For this reason, assume there exist distinct~$m_1, m_2, m_3 \in \targetMonomials$ such
  that~$m_{1} \cap m_{2} \cap m_{3} \neq \varnothing$.
  If~$k > 3$, observe that there exists a proper subset of~$\{i_1,
  \dots i_k\}$ and~$\{m_1, \dots, m_k\}$ that fulfills
  Condition~\ref{DelPia2}, contradicting the minimality assumption on~$k$.
  Hence, if~$k > 3$, Property~\ref{MIP2} holds.

  If~$k = 3$, the elements~$i_1$ and~$i_3$ from~\ref{DelPia2} fulfill~$i_3
  \in (m_3 \cap m_1) \setminus m_2$ and~$i_2 \in (m_3 \cap m_2) \setminus m_1$.
  Hence, $m_3 \cap (m_1 \cup m_2)$ is a proper superset of both~$m_3 \cap
  m_1$ and~$m_3 \cap m_2$.
  Because $m_1 \cap m_2 \cap m_3 \neq \varnothing$ by assumption,
  \ref{MIP1} follows.
\end{proof}

By combining these two lemmata, we obtain that the inequalities presented
by Del Pia and Khajavirad are a complete linear description of~$\BMP_n(\targetMonomials)$
in its original space if and only if~$P(\poset)$ is an integral polytope.

\subsection{Construction of the Extended Formulation}

As mentioned above, our results give rise to an extended formulation for $\BMP_n(\targetMonomials)$ in case $P(\poset)$ is an integral polytope.
We now discuss the verification of this condition and the computation of $\poset$ algorithmically.

\begin{theorem}\label{thm:RIPalgorithm}
  Let~$\targetMonomials$ be a set of monomials in variables~$x_1,\dots,x_n$.
  Using the RAM model,

  \begin{enumerate}[label=(\alph{*})]
  \item\label{it:optComplexA} it can be verified in time polynomial
    in~$\card{\targetMonomials}$ and~$n$ whether~$\targetMonomials$
    satisfies the monomial intersection property.

  \item\label{it:optComplexB} if~$\targetMonomials$ has the monomial
    intersection property, Problem~\eqref{eq:polyMin} with coefficients~$a
    \in \Q^{\targetMonomials}$ can be solved in polynomial time.
  \end{enumerate}
\end{theorem}

\begin{proof}
  To rule out that \ref{MIP1} holds requires to iterate over all triples
  in~$\targetMonomials$ and computing as well as comparing their (pairwise)
  intersections, which can be done in polynomial time. Property \ref{MIP2}
  can be verified by constructing a graph that contains for each monomial
  in~$\targetMonomials$ a node as well as an edge~$\{t_1, t_2\}$ for each pair of
  distinct~$t_1$, $t_2 \in \targetMonomials$ with~$t_1 \cap t_2 \neq \varnothing$. Then, deciding whether \ref{MIP2}
  holds, reduces to deciding whether the constructed graph contains a
  cycle, which can be done in polynomial time by BFS. Thus,
  the statement in~\ref{it:optComplexA} holds.

  For part~\ref{it:optComplexB}, recall that the
  linearization~$\poset$ is acyclic if the monomial
  intersection property holds. Hence, Problem~\eqref{eq:polyMin} can be
  solved by solving the linear program~\eqref{EquationLinearization} corresponding
  to~$\poset$ by \cref{thm:simpleLinearizationProjectionIntegral}. This linear program
  has~$\card{\allMonomials^\star}$ variables
  and~$\orderO(n\card{\allMonomials^\star})$ constraints. Hence,
  \ref{it:optComplexB} follows if we can show
  that~$\card{\allMonomials^\star}$ is polynomial in~$n$
  and~$\card{\targetMonomials}$ and~$\poset$ can be computed
  in time polynomial in~$n$ and~$\card{\targetMonomials}$.
  To this end, we consider a single monomial~$m \in \allMonomials^\star$
  first.

  \begin{claim}
    The number of successors of~$m$ in~$D^\star$ is at most~$2\,\card{m} - 1$.
  \end{claim}

  \begin{claimproof}
    Since~$\poset$ is acyclic, the subgraph of~$D^\star$
    containing~$m$ and all of its successors is a directed tree. Because
    this structure holds recursively for all successors of~$m$, we can use
    induction to prove the claim. Obviously, the claim is true for the
    base case~$\card{m} = 1$.

    For the inductive step, we assume that for a monomial~$m'$
    with~$\card{m'} \leq k$ the number of successors satisfies
    $\card{\succ(m')} \leq 2\,\card{m'} - 1$. If we are given a
    momomial~$m$ with~$\card{m} = k + 1$, the child nodes of~$m$ form a
    partition of the set~$m$, because otherwise, $G(\succ(m))$ contains a
    cycle. Then, the number of successors of~$m$ can be bounded by
    \[
    \card{\succ(m)} = 1 + \sum_{m' \in \delta^+(m)} \card{\succ(m')}
    \leq 1 + \sum_{m' \in \delta^+(m)} (2\, \card{m'} - 1)
    =
    1 + 2 (k+1) - \card{\delta^+(m)}.
    \]
    Since every non-trivial monomial has at least two successors, this
    expression is at most $2k + 1$, proving the claim.
  \end{claimproof}

  As a consequence, an upper bound for the number of nodes in~$D^{\star}$ is
  given by adding these upper bounds for all~$m \in \targetMonomials$
  (since $\succ(\targetMonomials) = \allMonomials^\star$) to get
  $\sum_{m \in \targetMonomials} (2\, \card{m} - 1)
  \in \orderO(n \cdot \card{\targetMonomials})$.

  To see that~$\poset$ can be computed in time polynomial
  in~$n$ and~$\card{\targetMonomials}$, recall
  that~$D(\poset)$ corresponds to the Hasse diagram
  of~$\allMonomials^\star$, endowed with the subset relation as partial
  order. For this reason, the algorithm presented in~\cite{KaibelPfetsch2002} can be used to
  compute~$\poset$, see \cref{rem:PropertyIntersectionPoset}. Because the running time of this
  algorithm is polynomial in the output size, the linear
  program~\eqref{EquationLinearization} can be generated in time polynomial
  in~$n \cdot \card{\targetMonomials}$ and has size~$\orderO(n \cdot \card{\targetMonomials})$. Thus, the optimization
  problem~\eqref{eq:polyMin} can be solved in polynomial time.
\end{proof}

\subsection{Combinatorial Algorithm}

In the previous section we saw that if $\targetMonomials$ satisfies the monomial intersection property, then the acyclic linearization $\linearization^\star$ can be computed in polynomial time.
We now present a combinatorial algorithm that, given such an acyclic linearization $\linearization = (n, \allMonomials, \constraints)$, solves the optimization problem~\eqref{eq:polyMin} by means of dynamic programming.
We consider coefficients $a \in \Q^\allMonomials$, i.e., in addition to~\eqref{eq:polyMin} we allow costs on the singletons as well.

\begin{theorem}
  \label{TheoremCombinatorialAlgorithm}
  Given a linearization $\linearization = (n, \allMonomials, \constraints)$
  for which $G(D(\linearization))$ is acyclic and an objective vector $a
  \in \Q^{\allMonomials}$, the optimization problem~\eqref{eq:polyMin} can
  be solved in time $\orderO(|\allMonomials|^2)$ (in the RAM model).
\end{theorem}

\begin{proof}
Let $D = D(\linearization)$ be the digraph of $\linearization$ with node set $\allMonomials$ and arc set $A$ and let $G = G(D)$ be the underlying undirected graph with edge set $E$.
Clearly, if $G$ is not connected, then the optimization problem decomposes into subproblems that can be solved independently.
Hence, we can assume that $G$ is connected, i.e., it is a tree.
We can assume that $\properMonomials \neq \varnothing$ since otherwise the problem is trivial.

\DeclareDocumentCommand\upMonomials{}{\allMonomials^{\text{up}}}
\DeclareDocumentCommand\downMonomials{}{\allMonomials^{\text{down}}}

Let the \emph{root monomial} $r \in \properMonomials$ be a monomial that is not contained in any \AND-constraint, i.e., $r$ has in-degree $0$ in $D(\linearization)$.
The map $\alpha\colon \allMonomials \setminus \{r\} \to \allMonomials$ shall denote the direction towards $r$ by setting $\alpha(m)$ (for each $m \in \allMonomials \setminus \{r\}$) to be the second node on the unique $m$-$r$-path.
Our dynamic program will propagate information within $G$ towards $r$.
We partition the set of non-root monomials into two sets $\upMonomials$ and $\downMonomials$ as follows.
Let $\upMonomials \coloneqq \{ m \in \allMonomials \setminus \setdef{r} \suchthat (\alpha(m),m) \in A \}$
and $\downMonomials \coloneqq \{ m \in \allMonomials \setminus \setdef{r} \suchthat (m,\alpha(m)) \in A \}$.

\DeclareDocumentCommand\zeroUp{}{\text{zero}^{\text{up}}}
\DeclareDocumentCommand\zeroDown{}{\text{zero}^{\text{down}}}
\DeclareDocumentCommand\oneUp{}{\text{one}^{\text{up}}}
\DeclareDocumentCommand\oneDown{}{\text{one}^{\text{down}}}

We now define the information that is propagated through the dynamic program.
For a node $m \in \allMonomials \setminus \{r\}$, define $\allMonomials(m)$ to be the node set of the connected component of $G \setminus \{\{m,\alpha(m)\}\}$ that contains $m$.
Define $\allMonomials(r) \coloneqq \allMonomials$.
For monomial subsets $\allMonomials' \subseteq \allMonomials$ define $Y(\allMonomials') \subseteq \{0,1\}^{\allMonomials'}$ as the set of vectors that satisfy all constraints of $\linearization$ that only involve monomials in $\allMonomials'$.
Moreover, we write $a(y) \coloneqq \sum_{m \in \allMonomials'} a_m y_m$ for vectors $y \in \R^{\allMonomials'}$.
Let $\zeroUp,\oneUp\colon \upMonomials \cup \{r\} \to \{0,1\}^\star$ be defined such that
\begin{align*}
  \zeroUp(m) & \in \arg\min\big\{ a(y) \suchthat y \in Y(\allMonomials(m)), ~y_m = 0 \big\} \text{ and} \\
  \oneUp(m) & \in \arg\min\big\{ a(y) \suchthat y \in Y(\allMonomials(m)), ~y_m = 1 \big\}
\end{align*}
hold.
Let $\zeroDown,\oneDown\colon \downMonomials \to \{0,1\}^\star$ be defined such that
\begin{align*}
  \zeroDown(m) & \in \arg\min\big\{ a(y) \suchthat y \in \{0,1\}^{\allMonomials(m)} \text{ and } (y,0) \in Y(\allMonomials(m) \cup \{\alpha(m)\} ) \big\} \text{ and} \\
  \oneDown(m) & \in \arg\min\big\{ a(y) \suchthat y \in \{0,1\}^{\allMonomials(m)} \text{ and } (y,1) \in Y(\allMonomials(m) \cup \{\alpha(m)\} ) \big\}
\end{align*}
hold, where the last coordinate in $(y,0)$ and $(y,1)$ shall correspond to monomial $m$.
Note that in particular the coefficient $a_{\alpha(m)}$ does not play a role in the definitions of $\zeroDown$ and $\oneDown$.

Clearly, the computation of $\zeroUp(r)$ and $\oneUp(r)$ solves the problem since the optimum value is $\min \{ a(\zeroUp(r)), a(\oneUp(r)) \}$.

We now describe how $\zeroUp$, $\oneUp$, $\zeroDown$ and $\oneDown$ can be computed recursively.
To this end, observe that $\zeroUp(m) = 0$ and $\oneUp(m) = 1$ holds for all $m \in \allMonomials$ with $\allMonomials(m) = \{m\}$.
The other nodes are considered in order of decreasing distance to $r$ (in $G$), which ensures that $\zeroUp$ and $\oneUp$ or $\zeroDown$ and $\oneDown$ are known for each of the neighbors with a larger distance to $r$.
Let $m \in \allMonomials$ be such a monomial, which may be the resultant of an \AND-constraint $\constraint = \{m_1, \dotsc, m_k \}$.
If this is not the case, e.g., if $m$ is a singleton that is contained in two \AND-constraints, then we let $\constraint = \varnothing$ and $k = 0$.
Moreover, let $\bar{\constraint}_1, \bar{\constraint}_2, \dotsc, \bar{\constraint}_\ell \in \constraints$ be those constraints in which $m$ is contained and whose resultant is not $\alpha(m)$.
Let $\bar{m}_i \coloneqq \bigcup \bar{\constraint}_i$ for $i \in [\ell]$ be their respective resultant monomials.
Clearly, $\bar{m}_i \in \downMonomials$ holds for all $i \in [\ell]$.
We distinguish two cases, depending on the propagation direction of $m$.

\medskip

\noindent
\textbf{Case 1:} $m \in \upMonomials \cup \{r\}$. \\
The vector $\oneUp(m)$ is constructed as follows.
First, set $\oneUp(m)_m \coloneqq 1$.
Second, set $\oneUp(m)_{\allMonomials(\bar{m}_i)} \coloneqq \oneDown(\bar{m}_i)$ for $i \in [\ell]$.
Third, set $\oneUp(m)_{\allMonomials(m_i)} \coloneqq \oneUp(m_i)$ for $i \in [k]$.
By induction, we only have to verify feasibility with respect to the constraints that involve $m$ (and do not involve $\alpha(m)$).
Clearly, if $\constraint \neq \varnothing$, then this constraint is satisfied since all participating coordinates of $\oneUp(m)$ are equal to $1$.
By definition of $\oneDown$, for $i \in [\ell]$, constraint $\bar{\constraint}_i$ is satisfied by $\oneDown(\bar{m}_i)$ extended with a $1$-entry at index $m$, and thus also by $\oneUp(m)$.
By induction, $a(\oneUp(m))$ is minimum.

The vector $\zeroUp(m)$ is constructed as follows.
First, set $\zeroUp(m)_m \coloneqq 0$.
Second, set $\zeroUp(m)_{\allMonomials(\bar{m}_i)} \coloneqq \zeroDown(\bar{m}_i)$ for $i \in [\ell]$.
As above, this ensures that, for $i \in [\ell]$, constraint $\bar{\constraint}_i$ is satisfied since $\zeroDown(\bar{m}_i)$ extended with a $0$-entry at index $m$ satisfies this constraint.
Third, for $i \in [k]$, an optimal choice for $\zeroUp(m)_{\allMonomials(m_i)}$ is determined by choosing the better (with respect to $a$) among $\zeroUp(m_i)$ and $\oneUp(m_i)$.
However, if $\constraint \neq \varnothing$, we have to be careful if $a(\oneUp(m_i)) \leq a(\zeroUp(m_i))$ holds for all $i \in [k]$, since a solution with $y_{m_i} = 1$ for all $i$ and $y_m = 0$ clearly violates constraint $\constraint$.
In this case it is feasible and optimal to select $\zeroUp(m)_{\allMonomials(m_i)} \coloneqq \zeroUp(m_i)$ for some index $i \in [k]$ for which $a(\zeroUp(m_i)) - a(\oneUp(m_i))$ is minimum (and $a(\oneUp(m_i))_{\allMonomials(m_j)} \coloneqq \oneUp(m_j)$ for all other indices).

\medskip

\noindent
\textbf{Case 2:} $m \in \downMonomials$. \\
Let $I \coloneqq \{ i \in [k] \suchthat m_i \neq \alpha(m) \}$ index those monomials in $\constraint$ of which we have computed $\zeroUp$ and $\oneUp$.
To construct the vector $\oneDown(m)$, we have to decide whether $\oneDown(m)_m$ should be equal to $0$ or equal to $1$.
To this end, we construct a candidate vector for each case.

We set $y^0_m \coloneqq 0$ and $y^0_{\allMonomials(\bar{m}_i)} \coloneqq \zeroDown(\bar{m}_i)$ for $i \in [\ell]$.
For $i \in I$, an optimal choice for $y^0_{\allMonomials(m_i)}$ is determined by choosing the better (with respect to $a$) among $\zeroUp(m_i)$ and $\oneUp(m_i)$.
However, we have to be careful if $a(\oneUp(m_i)) \leq a(\zeroUp(m_i))$ holds for all $i \in I$ since a solution with $y^0_{m_i} = 1$ for all $i \in I$ and $y^0_m = 0$, extended with a $1$-entry at index $\alpha(m)$, clearly violates constraint $\constraint$.
In this case it is optimal (and feasible) to select $y^0_{\allMonomials(m_i)} \coloneqq \zeroUp(m_i)$ for some index $i \in I$ for which $a(\zeroUp(m_i)) - a(\oneUp(m_i))$ is minimum (and $a(\oneUp(m_i))_{\allMonomials(m_j)} \coloneqq \oneUp(m_j)$ for all other indices).

We set $y^1_m \coloneqq 1$, $y^1_{\allMonomials(\bar{m}_i)} \coloneqq \oneDown(\bar{m}_i)$ for $i \in [\ell]$, and $y^1_{\allMonomials(m_i)} \coloneqq \oneUp(m_i)$ for $i \in I$.
By induction, we only have to verify feasibility with respect to the constraints that involve $m$.
Clearly, $\constraint$ is satisfied by $y^1$ extended with a $1$-entry at index $\alpha(m)$ since all participating coordinates are equal to $1$.
By definition of $\oneDown$, for $i \in [\ell]$, constraint $\constraint_i$ is satisfied by $\oneDown(\bar{m}_i)$ extended with a $1$-entry at index $m$, and thus also by $y^1$.
By induction, $a(\oneUp(m))$ is minimum (among those vectors with $y_m = 1$).

We can now select $\oneDown(m) \in \argmin\{ a(y^0), a(y^1) \}$.

The vector $\zeroDown(m)$ is constructed as follows.
First, set $\zeroDown(m)_m \coloneqq 0$.
Second, set $\zeroDown(m)_{\allMonomials(\bar{m}_i)} \coloneqq \zeroDown(\bar{m}_i)$ for $i \in [\ell]$.
This again ensures that, for $i \in [\ell]$, constraint $\constraint_i$ is satisfied since $\zeroDown(\bar{m}_i)_2$ extended with a $0$-entry at index $m$ satisfies this constraint.
Third, for all $i \in I$, the optimal choice for $\zeroDown(m)_{\allMonomials(m_i)}$ is determined by choosing the better (with respect to $a$) among $\zeroUp(m_i)$ and $\oneUp(m_i)$.
Constraint $\constraint$ is satisfied by the vector $\zeroDown(m)$ extended with a $0$-entry at index $\alpha(m)$ since $\zeroDown(m)_m = 0$.

This concludes the case distinction, showing how a vector $y \in Y(\allMonomials)$ with mininum $a(y)$ can be constructed.
To bound the running time, observe that we can account for all computations involving some vector $y \in \{\zeroUp(m), \oneUp(m), \zeroDown(m), \oneDown(m) \}$ at the step in which it is constructed.
This yields a running time of $\orderO(\card{\allMonomials})$ per propagation step.
Since less than $\card{\allMonomials}$-many such steps are required, the quadratic running time follows.
\end{proof}

Using our knowledge on the construction of $\linearization^\star$ for a given set $\targetMonomials$ that satisfies the monomial intersection property, we obtain the following strengthening of \cref{thm:RIPalgorithm}~\ref{it:optComplexB}.

\begin{corollary}
  Let~$\targetMonomials$ be a set of monomials in variables~$x_1,\dots,x_n$ that satisfy the monomial intersection property.
  Then, Problem~\eqref{eq:polyMin} with coefficients $a\in
  \Q^{\targetMonomials}$ can be solved in time $\orderO(\min(n^3 \cdot
  |\targetMonomials|^2, n^2 \cdot |\targetMonomials|^3))$ (in the RAM model).
\end{corollary}

\begin{proof}
  The problem can be solved by computing $\linearization^\star$ and running the algorithm from \cref{TheoremCombinatorialAlgorithm}.
  As in the proof of \cref{thm:RIPalgorithm}, the number of monomials of $\linearization^\star$ is at most $n \cdot |\targetMonomials|$.
  This leads to a running time of $\orderO(n^2 \cdot |\targetMonomials|^2)$ for the algorithm from \cref{TheoremCombinatorialAlgorithm}.
  However, this is dominated by the running time $\orderO(\min(n^3 \cdot |\targetMonomials|^2, n^2 \cdot |\targetMonomials|^3))$ that is required to compute $\linearization^\star$~\cite{KaibelPfetsch2002}.
\end{proof}

\section{Outlook}

The focus of this paper is on simple linearizations.
The only direct implication of our results on non-simple linearizations~$\linearization = (n, \allMonomials, \constraints)$, i.e., those for which proper monomials may be the resultant of several \AND-constraints, is the following sufficient criterion for deciding whether~$P(\linearization)$ is integral:
we can think of the set~$\constraints$ as the union of sets~$\constraints_i$ for $i \in \mathcal{I}$, corresponding to several simple linearizations~$\linearization_i$.
These simple linearizations can be found by selecting one \AND-constraint for every proper monomial being the resultant of at least two \AND-constraints.
We thus have
\[
  \proj[\singletonMonomials \cup \targetMonomials][P(\linearization)]
  =
  \proj[\singletonMonomials \cup \targetMonomials]\bigg(\bigcap_{i \in \mathcal{I}} P(\linearization_i)\bigg)
  \subseteq
  \bigcap_{i \in \mathcal{I}} \proj[\singletonMonomials \cup \targetMonomials][P(\linearization_i)].
\]
Consequently, if one of the simple linearizations~$\linearization_i$ is integral, also~$\proj[\singletonMonomials \cup \targetMonomials][P(\linearization)]$ is integral.
If the set of target monomials, however, does not have the monomial intersection property, no integral simple linearization exists.
An interesting question is whether there exists an integral non-simple linearization for a set $\targetMonomials$ not having the monomial intersection property.

More generally, it is interesting to investigate the strength of (simple or non-simple) linearizations in case no integral linearization exists.
Clearly, applying the linearization $\linearization^\star$ from \cref{sec:constructingAcyclicLinearizations} intuitively seems to be a good idea, but there is no clear dominance relationship if $\targetMonomials$ does not satisfy the monomial intersection property.
To see this, consider the instance with $n=3$ and $\targetMonomials = \{ \{1,2\}, \{2,3\}, \{1,2,3\} \}$, depicted in \cref{fig:hasseBad}.

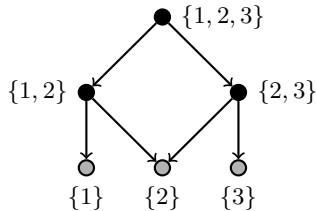
\begin{figure}[htb]
  \centering
  \begin{tikzpicture}
    \node[nodeSingleton,label=below:\footnotesize{$\{1\}$}] (1) at (1,0) {};
    \node[nodeSingleton,label=below:\footnotesize{$\{2\}$}] (2) at (2,0) {};
    \node[nodeSingleton,label=below:\footnotesize{$\{3\}$}] (3) at (3,0) {};

    \node[nodeTarget,label=left:\footnotesize{$\{1,2\}$}] (12) at (1,1) {};
    \node[nodeTarget,label=right:\footnotesize{$\{2,3\}$}] (23) at (3,1) {};
    \node[nodeTarget,label=right:\footnotesize{$\{1,2,3\}$}] (123) at (2,2) {};

    \draw[arc] (123) -- (12);
    \draw[arc] (123) -- (23);
    \draw[arc] (12) -- (1);
    \draw[arc] (12) -- (2);
    \draw[arc] (23) -- (2);
    \draw[arc] (23) -- (3);
  \end{tikzpicture}
  \caption{Example for which $P(\linearization^\star)$ does not dominate all relaxations that arise from simple linearizations.}
  \label{fig:hasseBad}
\end{figure}

The relaxation $P(\linearization^\star)$ does not imply inequality~\eqref{eq:linearizationSum} for $\constraint = \{ \{1,2\}, \{3\} \}$, which reads $y_{\{1,2\}} + y_{\{3\}} \leq y_{\{1,2,3\}} + 1$.
Hence, it is interesting to study how to select a good simple linearization depending on the given polynomial.

\medskip

\paragraph{\textbf{Acknowledgments}.}
We thank Monique Laurent for fruitful discussions.

\bibliographystyle{siam}
\bibliography{polynomialLinearization}

\end{document}